\documentclass[
reprint,
 aps,
 amsmath,amssymb,
]{revtex4-1}

\usepackage{graphicx, hyperref}
\usepackage{amsmath}
\usepackage{amssymb}
\usepackage{slashed}
\usepackage{mathrsfs}
\usepackage{yfonts}
\usepackage[T1]{fontenc}
\usepackage[utf8]{inputenc}
\usepackage[english]{babel}
\usepackage{amsthm}
\usepackage{tikz-cd}
\usepackage{braket}
\usepackage{physics}

\hypersetup{
pdfstartview = {FitH},
}
\hypersetup{
	colorlinks=true,         
	linkcolor=brown,          
	citecolor=red,        
	urlcolor=blue            
}

\theoremstyle{theorem}
\newtheorem{theorem}{Theorem}[section]
\newtheorem{lemma}{Lemma}[section]

\newtheorem{proposition}{Proposition}[section]

\theoremstyle{remark}
\newtheorem{remark}{Remark}[section]

\theoremstyle{definition}
\newtheorem{definition}{Definition}[section]

\usepackage{graphicx}% Include figure files
\usepackage{dcolumn}% Align table columns on decimal point
\usepackage{bm}% bold math
\numberwithin{equation}{section}

\usepackage{mathptmx}
\usepackage{etoolbox}

\makeatletter
\def\@email#1#2{%
 \endgroup
 \patchcmd{\titleblock@produce}
  {\frontmatter@RRAPformat}
  {\frontmatter@RRAPformat{\produce@RRAP{*#1\href{mailto:#2}{#2}}}\frontmatter@RRAPformat}
  {}{}
}%
\makeatother

\begin{document}

%\preprint{AIP/123-QED}

\title{Lieb-Schultz-Mattis Theorem and the Filling Constraint}
% Force line breaks with \\
\author{H. Chen}
\altaffiliation{email: hank.chen@uwaterloo.ca, \quad \& \quad ORCID: https://orcid.org/0000-0003-0504-1592}
 \affiliation{Department of Applied Mathematics, University of Waterloo, Ontario, Canada, N2L 3G1.}%Lines break automatically or can be forced with \\

\date{\today}% It is always \today, today,
 
\begin{abstract}
Following recent developments in the classification of bosonic short-range entangled (SRE) phases, we examine many-body quantum systems whose ground state fractionalization obeys the Lieb-Schultz-Mattis (LSM) theorem. We generalize the topological classification of such phases by {\it LSM anomalies} to take magnetic and non-symmorphic lattice effects into account, and provide direct computations of the LSM anomaly in specific examples. We show that the anomaly-free condition coincides with established filling constraints, and we also derived new ones on novel crystalline quantum systems.
\end{abstract}
            
\maketitle

\tableofcontents

\section{Introduction}\label{sec:nint}
It was established in the early 2000's that the topology of {\it non-interacting}/free quantum systems is protected by the band structure, for which the Altland-Zirnbauer-Kitaev classification scheme [\onlinecite{AZ}], [\onlinecite{Kit}] provided a complete classification based on Clifford modules/$K$-theory [\onlinecite{ABS}], [\onlinecite{Atiy}].

It had been successfully applied to classify integer Quantum Hall states, as well as the chiral $p$-wave superconductor in its weak-coupling limit [\onlinecite{BN}], for instance. Furthermore, the bulk-edge correspondence --- namely bulk topological invariant counting the number of chiral gapless edge states --- also finds rigorous ground in the celebrated Atiyah-Patodi-Singer index theorem. 

The above $K$-theory framework has been adapted to take the point group symmetry of the lattice into account [\onlinecite{MF}]. A generic $d$-dimensional lattice, however, has {\it space group} symmetry that need not split into discrete lattice translations and point group rotations, in which case it is called {\bf non-symmorphic}. On such a non-symmorphic glide lattice, it was show in the non-interacting context [\onlinecite{GT}], [\onlinecite{GT1}] that a non-trivial $\mathbb{Z}/2$-{\it Klein bottle phase} can be realized. Surprisingly, non-symmorphic effects also make an appearance in string theory [\onlinecite{SSG}], [\onlinecite{GT}], [\onlinecite{GT1}], [\onlinecite{AP}], [\onlinecite{DP}].

Despite the success of the $K$-theoretic approach, it requires crucially the non-interacting assumption. As such one cannot directly apply the $K$-theory classification scheme to describe, for instance, the fractional Quantum Hall effect. Some work had been done to achieve rational (orbifold) indices [\onlinecite{MM}], [\onlinecite{MT}] in the $K$-theory framework, but they nevertheless rest on fundamental assumptions about the effective behaviour of particles in the system. Thankfully, there had been alternative, inherently many-body descriptions of the integer and fractional Quantum Hall effect [\onlinecite{LRO}], [\onlinecite{BBRFr}], [\onlinecite{WPVZ}], [\onlinecite{AS}], [\onlinecite{BMNS}], [\onlinecite{SR}] based on the {\it filling index}, at zero-temperature. The counting of these filling indices play an important role in the {\bf Lieb-Schultz-Mattis (LSM) theorem}, which prohibits the existence of a gapped, symmetric non-degenerate ground state in the presence of translational symmetry.

Our goal in this paper is to study the effects of non-symmorphic crystalline symmetries and the presence of an anomalous magnetic flux on phases of gapped interacting quantum systems, by leveraging the classification of symmetry-protected bosonic invertible topological phase of matter by group cohomology classes [\onlinecite{CGW}], [\onlinecite{CGLW}], [\onlinecite{GJF}]: the so-called "in-cohomology" classification. 

The layout of the paper is the following: the LSM theorem and its relation to ground state degeneracy (GSD) will be described in Sec. \ref{sec:lsm}, then we summarize the topological theory [\onlinecite{ET}]. Next, we incorporate separately the magnetic and non-symmorphic twists in Sec. \ref{sec:lsmfill} into the topological classification. This allows us to compute crystalline and magnetic LSM anomalies, and reproduce filling constraints in Quantum Insulator (QI) and Quantum Hall (QH) systems that have been previously derived.
% [\onlinecite{BBRF}, [\onlinecite{BBRFr} [\onlinecite{PWJZ}, [\onlinecite{WPVZ}, [\onlinecite{LRO}. 
We then marry the two effects and lift the magnetic LSM anomaly to the non-symmorphic lattice in Sec. \ref{sec:nsymmag}, by generalizing the decorated domain wall construction [\onlinecite{RL}]. The LSM anomalies are then used to find new filling constraints that have no prior precedence. In the Appendixes, we describe the classification of 2D and 3D space groups [\onlinecite{Hill}], [\onlinecite{Dav}]  and compute their cohomology by use of the Lyndon-Serre-Hochschild spectral sequence. We consider a number of non-symmorphic cases that are of central interest to us.

\section{Lieb-Schultz-Mattis-Type Theorems}\label{sec:lsm}
It is well-known that the symmetry-protected topological (SPT) phases of a gapped quantum system on a fixed-dimensional space $X$ can be extracted from the symmetry group $G$ of its many-body ground states $\rho$ [\onlinecite{CGW}], [\onlinecite{CGLW}], [\onlinecite{KL}], [\onlinecite{GJF}], [\onlinecite{FH}]. A key result toward the inverse --- namely extrapolating the ground state given its phase --- is the {\bf Lieb-Schultz-Mattis (LSM) Theorem} [\onlinecite{LSM}]:
\begin{theorem}
\label{thm:lsm}
Given a quantum spin chain (i.e. a half-filled system of spin-$1/2$'s on a translationally-symmetric 1D lattice), the ground state is only one of the following:
\begin{enumerate}
\item ordered (spontaneously symmetry breaking), or
\item not gapped (gapless quantum spin liquid), or
\item fractionalized (degeneracy),
\end{enumerate}
in the thermodynamic limit.
\end{theorem}
\noindent The LSM theorem was further generalized to higher-dimensional spin systems by Hastings  [\onlinecite{Hastings}] and Oshikawa [\onlinecite{Oshi}], and is a very powerful tool for extracting ground state information from knowledge of just the quantum phase. In the following, we explain how this can be done purely topologically. 

\begin{remark}
The existence of the thermodynamic limit is a subtle point, however. For gapped Quantum Hall (QH) phases under mild technical conditions, it has been shown that the adiabatic curvature $\kappa \in\mathbb{Z}$ is indeed quantized [\onlinecite{AS}], and does in fact lead to the Hall conductance $\sigma_\text{Hall}$ in the thermodynamic limit [\onlinecite{BBRF}], [\onlinecite{BMNS}].
\end{remark}

\subsection{Lattice Homotopy and Defect Localization}\label{sec:lhdl}
We achieve the desired topological characterization in two steps: first, we define a notion of an equivalence of the spatial quantum data on the lattice $\Lambda$ embedded in $X$. Second, we describe how the short-range entanglement of the ground state puts constraints on how the LSM theorem can be satisfied. The first is dubbed {\it lattice homotopy} [\onlinecite{ET}], [\onlinecite{ET1}], [\onlinecite{PWJZ}], and the second is dubbed {\it defect localization/symmetry fractionalization} [\onlinecite{PWJZ}], [\onlinecite{LRO}].

\subsubsection*{Lattice homotopy}
Symmetries $G$ of a gapped Hamiltonian $H$ will in general define the Hilbert spaces $\mathcal{H}$ as an irreducible {\it projective} representation (proj-irrep) space of $G$ [\onlinecite{ET}], [\onlinecite{ET1}], [\onlinecite{PWJZ}]. Projective representations are required, as only {\it rays} of the ground states may be accessed by measurements and a $U(1)$ phase ambiguity between degenerate states manifests. Since proj-irreps are determined by the exact sequence
\begin{equation}
1\rightarrow U(1)\rightarrow G' \rightarrow G\rightarrow 1\nonumber
\end{equation}
of a central $U(1)$ extension $G'$ of $G$, they are classified by a second group cohomology class $\omega\in H^2(G,U(1))$. Here we take as an assumption that the symmetry group $G$ for the gapped Hamiltonian $H$ decomposes $G = K\times P$ into an internal compact Lie group $K$ and a spatial isometry discrete group $P$. This allows us to utilize the K{\" u}nneth formula
\begin{equation}
H^n(G,U(1)) \cong \bigoplus_{p+q=n}H^p(P,H^q(K,U(1))). \label{eq:kun}
\end{equation}
Note that on the lattice, $G$ may in general be a semidirect product $K\rtimes P$ due to {\bf spin-orbit coupling} (SOC).

Suppose we identify $K\subset G $ as the on-site internal symmetry, then each on-site Hilbert space $\mathcal{H}_x\subset \mathcal{H}$ constitutes a proj-irrep space of $K$ classified by a projective class $\omega_x \in H^2(K,U(1))$ assigned to each lattice site $x\in \Lambda$. A lattice $\Lambda$ decorated with this structure is called an {\bf anomalous texture}.
\begin{definition}
\label{def:anommove}
Let $f: \Lambda\rightarrow X$ be the lattice inclusion, and let $\omega:\Lambda\rightarrow H^2(K,U(1))$. The {\bf anomaly moves} on $\Lambda$ constitute
\begin{eqnarray}
\text{1. {\it Fusion}}: &\quad& \omega_x \leftrightarrow \prod_{y\in f^{-1}x}\omega_y,\nonumber \\
\text{2. {\it Elimination}}: &\quad& \omega_x = 1 \implies \Lambda \leftrightarrow \Lambda\setminus \{x\}.\nonumber
\end{eqnarray}
Two anomalous textures are {\bf lattice homotopically equivalent} if they are related by a series of anomaly moves and homotopies of $f$ [\onlinecite{ET}], [\onlinecite{ET1}]; see Fig. \ref{fig:anom}.
\end{definition}

\begin{figure}[h]
\centering
\includegraphics[width=1\columnwidth]{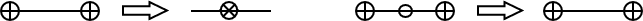}
\caption{The anomaly moves on a translationally-symmetric 1D anomalous texture with labels $\omega =\circ, \oplus,\otimes$ satisfying the bosonic fusion rule "$\oplus + \oplus = \otimes$", where $\circ$ is the unit/vacuum.}
\label{fig:anom}
\end{figure}

By the LSM {\bf Theorem \ref{thm:lsm}}, gapped $K$-symmetric ground states with non-integral filling exhibit emergent fractional excitations, and hence must be degenerate [\onlinecite{PWJZ}], [\onlinecite{WPVZ}]. In the thermodynamic limit, we wish to label classes of such symmetry fractionalizations with lattice homotopy classes $[\Lambda]$.
\begin{definition}
A trivial lattice homotopy class $[\Lambda]=0$ defines on $X\times\mathbb{R}$ an invertible $(d+1$)D {\bf anomaly-free} $K$-SPT phase based on the representation category $\operatorname{Rep}_\mathbb{C}(K)$ of $K$; such phases violate the LSM theorem.
\end{definition}
\noindent Together with the anomaly moves {\bf Definition \ref{def:anommove}}, $[\Lambda]$ encodes the excitations of a possibly anomalous invertible $P$-symmetric $(d+1)$D topological order $\mathcal{C}^{d+1}_P$ [\onlinecite{CGLW}], [\onlinecite{KL}], [\onlinecite{GJF}].

\begin{remark}\label{rmk:gspt}
The projective classes $\omega\in H^2(K,U(1))$ manifest as boundary anomalies of $(0+1)$D wires of (invertible) SPT phases [\onlinecite{Kit}], which are then stacked via $P$ to yield the order $\mathcal{C}^{d+1}_P$ [\onlinecite{ET}], [\onlinecite{RL}], [\onlinecite{XA}]. By holographic bulk-boundary correspondence [\onlinecite{LiWe}], a lattice homotopy class $[\Lambda]\simeq\Omega\mathcal{C}^{d+1}_P$ uniquely determines an anomaly-free SPT phase $\mathcal{D}^{d+2}_P\cong \operatorname{Bulk}\mathcal{C}_P^{d+1}$; in particular, $[\Lambda]=0$ means $\mathcal{C}^{d+1}_P$ is anomaly-free, and hence $\mathcal{D}_P^{d+2}=1^{d+2}$ is trivial.
\end{remark}

\subsubsection*{Defect localization}
Recall that the LSM theorem applies to microscopic theories on $\Lambda_0$, while the 't Hooft anomalies are data on the mesoscopic lattice $\Lambda$. In order to relate them and characterize the ground state degeneracy (GSD) topologically, we leverage the short-ranged nature of the ground state entanglement [\onlinecite{PWJZ}]:
\begin{definition}
Suppose the microscopic Hamiltonian $H = \sum\limits_{S\subset\Lambda_0}h(S)$ is a sum of local interactions where $S$ has finite (graph) width. Let $K\subset G$ be an on-site symmetry of $H$. A {\bf defect region} is a (connected) subset $Y\subset X_0$ for which $|Y|\leq \xi_0$ and 
\begin{equation}
H|_{Y} = \sum_{S\subset f^{-1}Y}h(S)\nonumber
\end{equation}
is not unitarily equivalent to $H$. The system exhibits {\bf defect/degeneracy localization} if $H$ hosts finitely many defect regions $Y$ for which the ground state space of $H|_Y$ is finite-dimensional.
\end{definition}
\noindent Once defect localization occurs, we may then assign projective phases $\omega_Y\in Z^2(K,U(1))$ determined by $K$ to label the degeneracies in each defect region $Y$. The total projective phase of the ground state decomposes
\begin{equation}
\omega_0 = \bigotimes_{\text{defect regions }Y}\omega_Y + O(e^{-r/\xi_0}) \nonumber
\end{equation}
up to an exponential suppression far past the correlation length $\xi_0>0$. By taking the thermodynamic limit $\Lambda_0\rightarrow\Lambda$ and coarse-graining $Y\rightarrow x$ to a single lattice site $x\in\Lambda$ (without breaking crystal lattice symmetry $P$), the cocycle $\omega_0$ then becomes a representative of the anomalous texture $\omega$ defining the lattice homotopy class $[\Lambda]$.

Based on this fact, the anomaly moves in {\bf Definition \ref{def:anommove}} can then be implemented via local unitaries [\onlinecite{CGLW}]. The algebra of such unitaries are referred to as the "fractionalization algebra" [\onlinecite{LRO}].
\begin{theorem}
\label{thm:kspt}
Suppose $P_0\subset P$ is a maximal Abelian subgroup with free action on $X$, and $P'=P/P_0$. Let $\mathcal{C}_P^{d+1}$ denote the first order lattice $K$-SPT phase determined by a lattice homotopy class $[\Lambda]\simeq \mathcal{C}^{d+1}_P$.
\begin{enumerate}
\item It is classified by the twisted orbifold cohomology $H^{d+2}_\text{orb}(BG,U(1)) = H^{d+2;\nu}_{P'}(P_0\times K,U(1))$ [\onlinecite{GT}], [\onlinecite{GT1}], [\onlinecite{MT}].
\item The LSM anomalies live in the contribution $H^{d}_\text{orb}(BP,H^2(K,U(1)))=H^{d;\nu}_{P'}(BP_0,H^2(K,U(1)))$ [\onlinecite{ET}].
\end{enumerate}
\end{theorem}
\begin{proof}
It is well-known [\onlinecite{CGLW}], [\onlinecite{GJF}], [\onlinecite{LiWe}], [\onlinecite{ET}], [\onlinecite{ET1}], [\onlinecite{RL}] that invertible bosonic $G$-protected topological phases in $(d+1)$D are characterized by $H^{d+2}(G,U(1))$. Given the decomposition $G=P\times K$ in the absence of spin-orbit coupling (SOC), we first demonstrate that this group becomes the twisted orbifold cohomology group.

By the classification of crystallographic groups [\onlinecite{Hill}], [\onlinecite{Dav}], there exists a maximal Abelian subgroup $P_0\subset P$ acting freely on $X$ such that $P$ is characterized by the central exact sequence
\begin{equation}
1 \rightarrow P_0\rightarrow P\rightarrow P'\rightarrow 1\nonumber
\end{equation} 
classified by a class $\nu\in H^2(P',P_0)$ [\onlinecite{SSG}], [\onlinecite{GT}], [\onlinecite{GT1}]. Moreover, the {\it point group} $P'=P/P_0$ is a finite group acting as local isotropies on $X$. Fix the space $X$ as a contractible universal covering $EP$, this leads to a tower of fibrations
\begin{equation}
\begin{tikzcd}
               &               &                                 & EP\cong X \arrow[d] \\
\ast \arrow[r] & BP' \arrow[r] & BP_0 \arrow[r] \arrow[ru, hook] & BP                 
\end{tikzcd}.\nonumber
\end{equation}
As $P'$ is finite, the space $BP_0$ is called the free universal Abelian cover [\onlinecite{MT}] --- or the {\bf fundamental domain} $\Gamma$ --- over the twisted toroidal nilorbifold $BP$ [\onlinecite{GT}], [\onlinecite{GT1}] --- or the {\bf unit cell} [\onlinecite{WPVZ}]. This defines the twisted orbifold cohomology $H^{d+2}(G,U(1)) \cong H^{d+2}_\text{orb}(BG,U(1)) = H^{d+2;\nu}_{P'}(BP_0\times K,U(1))$ [\onlinecite{GT}], [\onlinecite{GT1}], [\onlinecite{SSG}], [\onlinecite{MT}], which classifies the crystalline SPT phase $\mathcal{C}^{d+1}_P$ determined by a lattice homotopy class $[\Lambda]$ via the (bosonic) {\bf crystalline equivalance principle}. The mirror/rotation and Bieberbach results of lattice homotopy [\onlinecite{PWJZ}] are built into the descent $X\cong EP\rightarrow BP$.

Next we examine how to construct the LSM anomaly. Given a principal $K$-bundle $V$ on the spacetime cylinder $Y\times [0,1]$ of a defect region $Y\subset X_0$, its boundaries $\partial_\pm Y=Y\times\{0,1\}$ manifest the 2-cocycle $\omega_Y \in Z^2(K,U(1))$ through bundle maps $\chi_g: V|_{\partial_+Y}=V\rightarrow V=V|_{\partial_-Y}$ [\onlinecite{FH}], [\onlinecite{MuS}], where 
\begin{equation}
\chi_g\chi_{g'} = \omega_Y(g,g')\chi_{gg'}, \qquad g,g'\in K.\nonumber
\end{equation}
Take $\partial_- Y\cong Y\hookrightarrow X_0$ and send $\partial_+Y$ to temporal infinity, its class $\omega_Y$ then defines, in the thermodynamic limit $X_0\rightarrow X$, the requisite $(0+1)$D SPT wires characterized by $\omega\in H^2(K,U(1))$ on $X$. To stack these phases on $\Lambda$, let $\mathbb{Z}_o$ denote the orientation-graded $P$-module for which $g\cdot m = (\operatorname{det}g)m$ for all $m\in\mathbb{Z},g\in P$. Lattice sites, or the 0-cycle $c\in Z_0(P,\mathbb{Z}_o)$, are then sent to the class $\omega_*c\in Z_0(\Lambda,H^2(K,U(1)) \cong Z_0(P,H^2(K,U(1))$ via the coefficient map $\mathbb{Z}_o\mapsto \mathbb{Z}_o\cdot \omega$. We write its image as $Z_0(P,\langle\omega\rangle_o)$, where $\langle\omega\rangle\subset H^2(K,U(1))$ is the subgroup generated by a chosen projective class $\omega$.

In order to determine the lattice homotopy class $[\Lambda]$, however, we must build $P$-equivariance into $\omega_*c$. As a first order phase on the contractible $P$-symmetric space $X\cong EP$, we perform $d$-steps of {\bf homotopy descent} on the $0$-cycle ${\omega}_*c$ [\onlinecite{ET}], [\onlinecite{ET1}] to yield a $d$-cycle $\omega' \in Z_d(P,Z^2(K,U(1)))$. Its class $\omega'=\omega_*[\Lambda]$ then lies in the image of the fundamental class $[\Lambda]\in H_d(P,\mathbb{Z}_o)$ under the coefficient map. Now since the top degree group $H_d(X,\mathbb{Z}_o)$ is free, we have $(H_d)_P(X,\mathbb{Z}_o)\cong H^d_P(X,\mathbb{Z}_o)$ by the universal coefficient theorem. As such the class $\omega'$ uniquely determines a class $\tilde\omega\in H^d(P,H^2(K,U(1)))$.

This class $\tilde\omega$ is the sought-after LSM anomaly that classifies the $P$-symmetric topological order $\mathcal{C}^{d+1}_P$. Finally, as $H^\ast(K,U(1))$ is a trivial $P$-module in the absence of SOC, the $P'$-equivariant K{\" u}nneth formula states that $H^d(P,H^2(K,U(1)))\cong H^{d;\nu}_{P'}(BP_0,H^2(K,U(1)))$ is a subgroup of $H^{d+2;\nu}_{P'}(BP_0\times K,U(1))$.
\end{proof}
We stress that defect localization is the key property that allows a topological description of the LSM theorem [\onlinecite{CHR}], [\onlinecite{ET}]. Certain fermionic phases can violate defect localization [\onlinecite{PWJZ}], hence fermionic 't Hooft anomalies $\widehat{\Omega}^{d+1}_{\text{Spin}_c}(B(G\times\mathbb{Z}/2^F))$ [\onlinecite{GOPWW}] do not necessarily classify symmetry fractionalization classes [\onlinecite{CHR}], [\onlinecite{BBCW}]. Furthermore, it was shown [\onlinecite{Deb}], [\onlinecite{ZYGQ}] that fermionic crystalline equivalence must necessarily encode SOC $P\ltimes K$.

\subsection{Topological Theory of LSM Theorems}\label{sec:toplsm}
The main topological theory [\onlinecite{ET}] relies on the hypothesis that bosonic invertible topological phases of matter forms a spectrum $E=E_{\leq \ast}$ (up to a certain dimension) [\onlinecite{FH}], [\onlinecite{SP}], and the $G$-protected phases are classified by the generalized cohomology theory $\mathcal{E}_\ast=H^\ast(BG,E)$ [\onlinecite{GJF}], [\onlinecite{BC}]. The key ingredients useful for our applications are:
\begin{enumerate}
\item At dimension grading $l=1$, there exists a homomorphism
\begin{equation}
\pi_0 E_1 \rightarrow H^2(K,U(1))\nonumber
\end{equation}
associating a gapped (0+1)D $K$-SPT wire to the projective class on the lattice $\Lambda$.
\item The (dual) spectrum map $\Omega:{E}_{l+1} \rightarrow{E}_l$ induces a bulk-boundary map 
\begin{equation}
\pi_1 E_2 \rightarrow \pi_0 E_1\nonumber
\end{equation}
in the spectrum $E$ at degree $l =2$.
\end{enumerate}

\begin{remark}\label{rmk:fmt}
One proposal for $\mathcal{E}_\ast$ is the Borel-Moore cohomology $[MTG,\Sigma^{\ast+1}I\mathbb{Z}]$ based on the stable Madsen-Tillmann spectrum $MTG$ [\onlinecite{FH1}]. In free theories, $\mathcal{E}_\ast(X) \simeq K^{-\ast}(X) \cong [\Sigma^{-\ast}X,U]$ is the $K$-theory spectrum  [\onlinecite{KL}], [\onlinecite{SSG}], [\onlinecite{GT}], [\onlinecite{GT1}], [\onlinecite{BN}], [\onlinecite{Kit1}] equipped with the Chern character $\operatorname{ch}:K^\ast(X) \rightarrow H^\ast(X,\mathbb{Q})$ [\onlinecite{Atiy}]. Furthermore, if $\widehat{\Lambda}$ denotes the Pontjagyn dual of $\Lambda$, the {\bf $T$-duality} map [\onlinecite{MT}], [\onlinecite{Sch}]
\begin{equation}
\mathscr{T}:K^{\ast}(B\Lambda) \xrightarrow{\sim} K^{\ast}(\widehat{\Lambda}), \label{eq:tbc}
\end{equation}
as a composition of Poincar{\' e} duality with Baum-Connes assembly, preserves the topological order under a Fourier-Mukai transform [\onlinecite{SSG}], [\onlinecite{GT}], [\onlinecite{GT1}], [\onlinecite{BK}].
\end{remark}

\begin{theorem} \label{thm:lsmanom}{\bf (Els{\' e}-Thorngren, 2019).}
If the projective class $\omega\in H^2(K,U(1))$ embeds non-trivially into a {\bf LSM anomaly} $\tilde\omega\in H^d_{K}(X,\mathbb{Z})$, then $\tilde\omega$ defines a non-trivial lattice homotopy class $[\Lambda]\neq 0$ on $\Lambda\subset X$. On the other hand, if $\omega = d\tau$ is trivial, then $\tilde\omega = d\tilde\tau$ defines the defect network $\tilde\tau$ that trivializes the boundary phase on $\Lambda$, up to lattice homotopy equivalence.
\end{theorem}
\noindent In other words, $H^2(K,U(1))$ is an "adequate label" [\onlinecite{CGLW}] for the $(0+1)$D invertible $K$-SPT wires relevant for the LSM theorem. 

As such, we focus our attention only on the {\bf in-cohomology}  part [\onlinecite{RL}]
\begin{equation}
H^{d-l+2}(K,U(1)) \subset \lim_k \pi_{d-l+k}\Sigma^k E \nonumber
\end{equation}
of the stable homotopy spectrum of $E$, where $\Sigma:E_\ast\rightarrow E_{\ast+1}$ is the suspension map [\onlinecite{BC}]. By {\bf Theorem \ref{thm:kspt}} and the crystalline equivalence principle, first order $l=d$ crystalline invertible $K$-SPT phases in dimension $d$, denoted $\pi_0\mathcal{E}^\text{in}_0$, are then classified by elements in $H^d(P,H^2(K,U(1))) = H^{d;\nu}_{P'}(\Gamma,H^2(K,U(1)))$.

\subsubsection{Characteristic classes of the non-symmorphic lattice}\label{sec:eul}
The computation of classes in the symmorphic case $P\cong \mathbb{Z}^d\rtimes P'$ had been partially carried out [\onlinecite{ET}]. Given the embedding $\rho:P'\hookrightarrow O(d)$, the linear action of $P'$ is characterized by the pullback 
\begin{equation}
e(P') =\rho^*e_d \in H^d(P',\mathbb{Z}_o)\cong H^d_{P'}(X,\mathbb{Z}_o) \nonumber
\end{equation}
of the $O(d)$-{\bf Euler class} $e_d \in H^d(X,\mathbb{Z}_o)$ on $X$, with graded coefficients $\mathbb{Z}_o$. 

Symmorphicity of the lattice allows us to push this class $e(P')$ forward through the descent $X\rightarrow B\mathbb{Z}^d = \Gamma$ onto the fundamental domain. Denoting this $P'$-equivariant class again by $e(P')\in H^d_{P'}(\Gamma,\mathbb{Z}_o)$, its order $|e(P')|$ counts the number of {\bf Wyckoff positions} in the fundamental domain $\Gamma$. Examples of $e(P')$ are computed in Appendix \ref{sec:lhs}.

We do not have the luxury of the Euler class, however, in the non-symmorphic case. Nevertheless, given a space group $P$ characterized by an extension class $\nu\in H^2(P',\mathbb{Z}^d)$, we may combine a characteristic class $e_\nu \in H^d(P,\mathbb{Z}_o)$ and the anomalous texture $\omega\in H^2(K,U(1))$ via the coefficient map
\begin{equation}
\omega_*:\mathbb{Z}_o \rightarrow H^2(K,U(1)_o), \qquad m\mapsto m\cdot \omega \label{eq:coeff}
\end{equation}
to give the LSM anomaly 
\begin{eqnarray}
\omega_*(e_\nu) &\in& H^{d}_\text{orb}(BP,H^2(K,U(1))) \nonumber \\
&\quad& \cong H^{d;\nu}_{P'}(\Gamma,H^2(K,U(1))) \nonumber
\end{eqnarray}
that classifies the first order crystalline $K$-SPT phase $\mathcal{C}^{d+1}_P$ determined by a lattice homotopy class $[\Lambda]$ [\onlinecite{ET}], [\onlinecite{ET1}], [\onlinecite{RL}]. 

By leveraging this idea, we shall assemble crystalline LSM anomalies and produce versions of the LSM theorem that take into account non-symmorphic and magnetic effects. We will also reproduce the filling constraints derived previously [\onlinecite{LRO}], [\onlinecite{BBRF}], [\onlinecite{WPVZ}].

\subsubsection{Magnetic and non-symmorphic extensions}
Consider the case where $K_f\subset K$ is a central subgroup extending $K_p$ to $K$, and the projective class $\omega\in H^2(K_p,U(1))\subset H^2(K,U(1))$ is determined purely by $K_p$. Non-symmorphic and magnetic effects are characterized by {\it twists} $\nu \in H^2(P',\mathbb{Z}^d), \beta\in H^2(\mathbb{Z}^d,K_f)$ [\onlinecite{SSG}], [\onlinecite{GT}], [\onlinecite{GT1}], [\onlinecite{MT}] which classify exact sequences
\begin{eqnarray}
\begin{tikzcd}
1 \arrow[r] & \mathbb{Z}^d \arrow[r] & P \arrow[r] & P' \arrow[r] \arrow[l, "\widehat{\cdot}"', bend right,dashed] & 1
\end{tikzcd},\nonumber \\ 
\begin{tikzcd}
1 \arrow[r] & K_f \arrow[r] & \mathbb{Z}^d_\text{flux} \arrow[r] & \mathbb{Z}^d \arrow[r] \arrow[l, "\tilde{\cdot}"', bend right,dashed] & 1
\end{tikzcd}\nonumber
\end{eqnarray}
whose sections $\hat{\cdot}:P' \rightarrow P,\tilde{\cdot}:\mathbb{Z}^d\rightarrow \mathbb{Z}^d_\text{flux}$ see group-theoretic obstructions
\begin{eqnarray}
{[\widehat{a},\widehat{b}]}= \widehat{ab}\hat{b}^{-1}\hat{a}^{-1} = \nu(a,b),&\quad& \forall a,b\in P', \nonumber \\
{[\tilde{t},\tilde{s}]} =\widetilde{t+s}-\tilde{s}-\tilde{t}= \beta(t,s), &\quad& \forall t,s\in\mathbb{Z}^d \label{eq:obst}
\end{eqnarray}
by cocycle representatives of $\nu,\beta$. The extension $P$ defines the {\bf non-symmorphic space group} [\onlinecite{SSG}], while $\mathbb{Z}^d_\text{flux}$ defines the {\bf magnetic translation algebra} [\onlinecite{LRO}]. These cocycles $\nu,\beta$ can be described explicitly as follows:
\begin{itemize}
\item By virtue of $K$ acting on-site, the magnetic translation algebra must be $K$-invariant; this implies that $\beta(t,s)$ commute with all of $K$, and hence lie in a central subgroup $K_f$ for all $t,s\in\mathbb{Z}^d$. By identifying generators $\{t_\alpha\}_{\alpha\leq d} \in H^1(\mathbb{Z}^d,\mathbb{Z})$ with the primitive lattice vectors for $1\leq \alpha\leq d$, we may express
\begin{equation}
\beta = \sum_{\alpha<\beta }\varphi_{\alpha\beta}t_\alpha \cup t_\beta \nonumber
\end{equation}
in terms of a {\bf flux} $\varphi_{\alpha\beta}\in K_f$ through the plaquette spanned by translations $t_\alpha$ and $t_\beta$. 

\item By performing a Fourier transform $H^2(P',\mathbb{Z}^d) \rightarrow H^2(P',\mathbb{T}^d)$ [\onlinecite{GT}], [\onlinecite{GT1}], the image of $\nu$, also denoted by $\nu$, satisfies
\begin{equation}
\nu(a,b)\chi = ab\cdot \chi(\nu(a,b)) \in U(1), \nonumber
\end{equation}
where $a,b\in P'$ and $ \chi \in \widehat{\mathbb{Z}^d} \cong \mathbb{T}^d$. By Bieberbach's theorem, all extensions $P$ can be characterized by an action of $P'$ on $\mathbb{Z}^d$ twisted by $\nu$ [\onlinecite{GT}], [\onlinecite{GT1}]. Non-symmorphic elements of finite order in $P$ must then be either one of the two types: 
\begin{enumerate}
\item {\bf $p_q$-Screw/rototranslation}: If $a\in P'$ is of order $p> 1$ such that $\hat{a}^p = q \in \mathbb{Z}^d$, then 
\begin{equation}
\nu(a,a)(k) = e^{i \frac{q\cdot k}{p}}, \qquad k \in \mathbb{T}^d. \nonumber
\end{equation}
\item {\bf Glide/reflectotranslation}: If $a\in P'$ is a reflection such that $\hat{a}^2 = 1 \in \mathbb{Z}^d$, then
\begin{equation}
\nu(a,a)(k) = e^{-i\frac{1}{2}p_a(k)}, \qquad k \in \mathbb{T}^d \nonumber
\end{equation}
where $p_a: \mathbb{T}^d \rightarrow \mathbb{T}^1$ projects onto the reflection axis of $a$.
\end{enumerate}
Notice that $\nu$ is a translation in real space but only a phase in momentum space. The screws and glides are distinguished by their {\it grading} $c: P'\rightarrow\mathbb{Z}/2$ [\onlinecite{SSG}].
\end{itemize}
We shall mainly consider the split extension $K \cong K_f\rtimes K_p$, in which there is a clear distinction between the flux $\beta\in H^2(\mathbb{Z}^d,K_f)$ and the projective class $\omega\in H^2(K_p,U(1))$.

%In general, suppose we are given a $K$-module $A$ and a $P'$-module $B$. Let $\omega \in H^2(K_\text{prof},A), \sigma\in H^2(P',B)$ be two group cohomology classes characterizing the proj-irreps of $K_p$ and $P'$ respectively, we wish to lift them into classes $\tilde{\omega}\in H^2(K,A),\hat{\sigma}\in H^2(P,B)$ characterizing the proj-irreps of the central extensions $K$ and $P$, respectively. Following Ref. [\onlinecite{ET}, this can be accomplished by the lifting homomorphisms $\alpha,\mu$ in the following diagram
%\begin{center}
%\begin{tabular}{c c}
%\begin{tikzcd}
%{\tilde{\omega}\in H^2(K,A)} \arrow[d] &  & {H^1(K_f,A)} \arrow[ll] \arrow[lld, "\alpha^*", dashed] \\
%{\omega\in H^2(K_p,A)}                         &  &                                                                 
%\end{tikzcd} & \qquad 
%
%\begin{tikzcd}
%{\hat{\sigma}\in H^2(P,B)} \arrow[d] &  & {H^1(\mathbb{Z}^d,B)} \arrow[ll] \arrow[lld, "\mu^*", dashed] \\
%{\sigma\in H^2(P',B)}                         &  &                                                                 
%\end{tikzcd},
%\end{tabular}
%\end{center}
%which exist if the central extensions $K,P$ are trivial in cohomology. We shall use the left diagram in particular to 

\section{LSM Anomaly and the Filling Constraint}\label{sec:lsmfill}
Suppose the lattice $\Lambda$ hosts $N$-flavours of particles, each equipped with a number operator $n_a$, whose integrality $\operatorname{Spec}n_a\subset\mathbb{Z}$ is guaranteed by a global symmetry $U(1)_C^N$ of charge conservation. 
\begin{definition}
The {\bf filling factor} $\nu_a=\sum\limits_{x\in \Gamma}\expval{n_a(x)}$ is the average number of $a$-excitations per fundamental domain $BP_0=\Gamma$.
\end{definition}
\noindent It was demonstrated [\onlinecite{LRO}], [\onlinecite{WPVZ}] that the existence of a non-degenerate gapped symmetric ground state rests on constraints of the filling factor. This observation is made precise by the {\it filling invariant} [\onlinecite{CHR}] $$\sum_{a\leq N} \kappa_a\nu_a\mod \mathbb{Z}\in \mathbb{Z}/(\epsilon_NN),$$ characterizing part of the 't Hooft anomaly $\widehat{\Omega}^{2+1}_{\text{Spin}_c}(B\mathbb{Z}/N)\cong\mathbb{Z}/(\epsilon_NN)\times\mathbb{Z}/\frac{N}{\epsilon_N}$ where $\epsilon_N=\begin{cases}1 &; N=1\mod 2 \\ 2 &; N=2\mod 2\end{cases}$. The {\bf anomaly-free condition} \begin{equation}
\sum\limits_a \kappa_a\nu_a =0 \mod\mathbb{Z}\label{eq:anomfree}
\end{equation}
reproduces precisely these filling constraints. 

\begin{remark}
Here $\kappa_a$ is the {\it chiral anomaly} associated to the charge conservation symmetry, which is even $\kappa_a =0\mod 2$ if the $a$-particle is bosonic and odd $\kappa_a=1\mod 2$ if $a$ is fermionic [\onlinecite{CHR}]. As we shall mainly focus on bosonic theories, we understand $\nu_\omega$ as counting the number of "bosonic charges" in the fundamental domain $\Gamma$, and neglect the $\kappa_a=2$ factor.
\end{remark}

\subsection{LSM Anomaly of the Quantum Insulator}\label{sec:lsmqshe}
Consider first a system of spin-$\frac{1}{2}$'s on a translationally-symmetric lattice $\Lambda$, with $K_f=0$ and $K_p=SO(3)$. In this case, an element $\omega\in H^2(SO(3),U(1))$ classifies the choice of a projective lift 
\begin{equation}
1 \rightarrow U(1) \rightarrow \operatorname{Spin}^\pm_c \rightarrow SO(3) \rightarrow 1 \nonumber
\end{equation}
describing integer and half-integer spins, coming from the second Stiefel-Whitney class $w_2(SO(3)) \in H^2(SO(3),\mathbb{Z}/2)$ [\onlinecite{ABS}], [\onlinecite{FH}], [\onlinecite{ET}], [\onlinecite{SSG}]. Since translational symmetry is taken care of by the descent to the fundamental domain $\Gamma=B\mathbb{Z}^d$, the 2-torsion property of $\omega=w_2(SO(3))$ gives $N=k=2$ and the spin group $\langle \omega\rangle \cong \mathbb{Z}/2$ constitutes the totality of the particle-flavours.

If we define $\nu=\frac{1}{2}\left(\nu_{\frac{1}{2}} + \nu_1\right)$ as the average filling on $\Gamma$, then in the absence of TR symmetry, the anomaly-free condition Eq. (\ref{eq:anomfree}) implies the ordinary Quantum Insulator (QI) filling constraint $\nu\in \mathbb{Z}$ [\onlinecite{WPVZ}], [\onlinecite{BMNS}], [\onlinecite{SR}]. If $\nu_1=0$, then we must fully-fill the spin-half lattice $\nu_{\frac{1}{2}}\in2\mathbb{Z}$. This is precisely the classical LSM {\bf Theorem \ref{thm:lsm}}.

% \begin{remark}
% It is quite remarkable that one may describe half-integer spin particles with a completely bosonic theory, as these particles are fermionic by the spin-statistics theorem. The $U(1)_C$-connection associated to charge conservation symmetry is used essentially as a $\text{Spin}_c$-connection.
% \end{remark}

With TR symmetry $T\in\mathbb{Z}/2^T$, however, the ground state acquires a degeneracy introduced by Kramers pairs. This can be understood as the two $\mathbb{Z}/2^T$ extensions of spin $\langle w_2(SO(3))\rangle\cong\mathbb{Z}/2$, characterized by $H^2(\langle\omega\rangle,\mathbb{Z}/2^T)\cong \{\mathbb{Z}/2^2,\mathbb{Z}/4\}= \mathbb{Z}/2$ representing Kramers pairs. This replaces the particle-flavours $a$ with Kramers doublets $\text{Kr+}$ and singlets $\text{Kr-}$, which now carry two spins each. The anomaly-free condition Eq. (\ref{eq:anomfree}), $\nu_\text{Kr}= \frac{1}{2}\left(\nu_\text{Kr+}+\nu_\text{Kr-}\right)=\frac{1}{2}\nu\in \mathbb{Z}$, then implies that $\nu\in2\mathbb{Z}$ must be even [\onlinecite{WPVZ}].

Now if the lattice is symmorphic $P\cong \mathbb{Z}^d\rtimes P'$, we may perform the descent $X\rightarrow B\mathbb{Z}^d=\Gamma$ to the fundamental domain. The first order (0+1)D $K$-SPT phase is then defined by the choice of a projective class $\omega\in H^2(K,U(1))$ assigned to each $P'$-orbit in $\Gamma$. 

Physically speaking, the order $k\in\mathbb{Z}$ of the projective class $\omega$ determines the number of particle-flavours, and that of the Euler class $e(P')\in H^d(P',\mathbb{Z}_o)$ counts the number of {\bf Wyckoff positions} these particles can live. If $\nu_a$ denotes the filling of particle-type $a$ at each Wyckoff position, then the filling invariant can be obtained
\begin{equation}
p\nu_\omega\mod \mathbb{Z},\qquad \nu_\omega=\frac{1}{k}\sum\limits_{a\leq k}\nu_a\label{eq:symmorph}
\end{equation}
from the LSM anomaly $\omega_*(e(P'))$ via Eq. (\ref{eq:coeff}), with $p\in\mathbb{Z}$ the order of $e(P')$. In the decorated domain wall approach [\onlinecite{RL}], the projective classes $\omega$ live at the endpoints of the (0+1)D $K$-SPT phase.

\subsubsection{Lattice characteristics classes and non-symmorphic filling invariants}\label{sec:nsymfill}
If $P$ is non-symmorphic, we must compute directly the crystalline characteristic class $e_\nu\in H^d(P,\mathbb{Z}_o)$. This class controls how the projective classes $\omega$ are assigned throughout the lattice. Let us mainly focus on the non-symmorphic Bravais lattices in $d=2,3$, and quote results from Sec. \ref{sec:lhs}.

First, all 2D Bravais lattices with $P'\supset C_n$ must have $n=2,3,4,6$, and glides/reflectotranslation manifest only in case $n=2$. For other values, the space group is symmorphic and the filling invariant is reproduced by Eq. (\ref{eq:symmorph}) with $p=n$. Now the glide lattice $\Lambda=\mathbb{Z}\oplus\mathbb{Z}'$ (see Fig. \ref{fig:glides}), whose space group $P=\mathtt{pg}$ fits into the central extension sequence
\begin{equation}
1 \rightarrow \Lambda \xrightarrow{\operatorname{id}\oplus \cdot 2} \mathtt{pg}\xrightarrow{(-1)^{n_2}} C_2\rightarrow 1 \nonumber
\end{equation}
characterized by the unique non-trivial extension class $\nu\in H^2(C_2,\Lambda)\cong \mathbb{Z}/2$, has the cohomology Eq. (\ref{eq:glidecoh})
\begin{equation}
e_\nu =2\cdot 1\in H^2(P,\mathbb{Z}_o) \cong 2\mathbb{Z}.\nonumber
\end{equation}
Given a projective class $\omega\in H^2(K,U(1))$, we may construct the LSM anomaly $\omega_*(e_\nu) \in H^2(P,H^2(K,U(1))_o)$ via Eq. (\ref{eq:coeff}). 

The even-ness of $e_\nu$ doubles the assignment of $\omega$ to the fundamental domain, hence
\begin{equation}
\nu_\omega \mod 2\mathbb{Z}.\nonumber
\end{equation}
It is important to note that the glide reflection is represented as an antiunitary operator on $\mathbb{Z}/k$. For a TR-invariant QI with $k=2$, the anomaly-free condition reproduces the known [\onlinecite{WPVZ}] filling constraint $\nu_\omega \in 4\mathbb{Z}$.

In the 3D case, consider the axial rotational point group $P'=1\times C_n$ exhibiting a screw/rototranslation lattice $\Lambda$ (see Fig. \ref{fig:screws}) characterized by a divisor $A$ of $n$ for which $\operatorname{gcd}(\frac{n}{A},A)=1$. The non-symmorphic space group $P=P_{A;\nu}$ sits in the extension sequence
\begin{equation}
1\rightarrow\Lambda \rightarrow P_{A;\nu}\rightarrow 1\times C_n\rightarrow 1\nonumber
\end{equation}
classified by a non-trivial extension class $\nu\in H^2(1\times C_n,\Lambda)\cong \mathbb{Z}/\frac{n}{A}$.

In the generic case of an $A_{\frac{n}{A}}$-screw given by the generator $\nu=1$, the cohomology Eq. (\ref{eq:screwcoh})
\begin{equation}
e_\nu = (1,\frac{n}{A}\cdot 1)\in H^3(P_{A;\nu},\mathbb{Z}_o) \cong \mathbb{Z}/A \oplus \frac{n}{A}\mathbb{Z} \nonumber
\end{equation}
constructs the LSM anomaly $\omega_*(e_\nu) \in H^3(P_{A;\nu},H^2(K,U(1))_o)$ as previously. Physically, $\omega$ is assigned to each $A$-number of Wyckoff positions across the $\frac{n}{A}$ screw translations:
\begin{equation}
A\nu_\omega \mod \frac{n}{A}\mathbb{Z}.\nonumber
\end{equation}
Here, the $A$-fold screw rotation is unitary. Taking $p=A$ and $q=\frac{n}{A}$, the anomaly-free condition $p\nu_\omega \in kq \mathbb{Z}$ reproduces the filling constraint $\frac{\nu_\omega}{q}\in \frac{2}{p}\mathbb{Z}$ in the TR-invariant QI [\onlinecite{WPVZ}].

We now turn to a novel effect: the coexistence of glides and screws. Consider the chirorotational point group $P' = C_2\times C_6$ on the {\it unique} glide-screw lattice $\Lambda$ (see Fig. \ref{fig:glidescrew}). The space group $P=\mathtt{gs}$ is classified by the non-trivial extension class $\nu =(1,1)\in H^2(C_2\times C_6,\Lambda) \cong \mathbb{Z}/2\oplus\mathbb{Z}/2$. Here the representation $\rho:P'\rightarrow\operatorname{GL}(\Lambda)$ is given by the divisor $A'=3$ of $n=6$, such that $\frac{n}{A'}=2$ is even. 

The generator of the cohomology Eq. (\ref{eq:gscoh}),
\begin{equation}
e_\nu\in H^3(\mathtt{gs},\mathbb{Z}_o)\cong 2\mathbb{Z}\oplus \mathbb{Z}/3,\nonumber
\end{equation}
then assembles the LSM anomaly $\omega_*(e_\nu)$ as previously, which gives the filling invariant
\begin{equation}
3\nu_\omega \mod 2\mathbb{Z}.\nonumber
\end{equation}
There could in principle be $A=2A'$ number of Wyckoff positions to fill in the fundamental domain, but the coupling between the order-2 screw rotation and the glide reflection halves this availability. Furthermore, the glide-screw element in $\mathtt{gs}$ is antiunitary, hence anticommutes with TR.

\begin{remark}
Of course, if $\nu_\omega\not\in\mathbb{Z}$, then the filling constraint Eq. (\ref{eq:anomfree}) is violated regardless if the crystalline symmetry $P$ is present. Those systems with $\nu_\omega\in\mathbb{Z}$ but violates the filling constraint in the presence of crystalline symmetry are called {\bf weak crystalline phases} [\onlinecite{RL}], [\onlinecite{XA}]. It was shown [\onlinecite{XA}] that weak glide-symmetric phases square/double-stack to the trivial phase, which is captured precisely by the doubling in the filling factor $\nu_\omega$. One may also understand this through the fact that $B\mathtt{pg}=\mathbb{K}$ is the Klein bottle [\onlinecite{SSG}], and $\mathbb{T}^2$ is a double cover of $\mathbb{K}$.
\end{remark}

\subsubsection{Lattice reduction*}
Let us now turn to an alternative approach for computing the filling constraints [\onlinecite{WPVZ}]. In the axial case, $\nu$ is $\mathbb{Z}$-valued, say $q t_1$, along a specific axis $\Lambda_1=\mathbb{Z}[t_1]\subset \Lambda$, hence we may define a {\bf lattice reduction map} $\mu:\Lambda \rightarrow \Lambda_\nu$ such that $\mu\circ \nu=0$ vanishes. The extension sequence along the induced surjection $\tilde{\pi}:\mu P\rightarrow P'$ is thus split, and we may lift the Euler class $e(P')\in H^d(P',\mathbb{Z}_o)$ of the point group $P'$ to $e_{\slashed\nu}=\tilde{\pi}^*e(P') \in H^d(\mu P,\mathbb{Z}_o)$. 

This gives a class $e_{\slashed\nu} \in H^d_{P'}(\Gamma_{\slashed\nu},\mathbb{Z}_o)$ on the commutative space $\Gamma_{\slashed\nu} = B\Lambda_\mu$. Physically, this procedure restores the full point group symmetry $P'$ by reducing the fundamental domain $B\mu:\Gamma\rightarrow \Gamma_{\slashed\nu}$. One may then take $B\mu^* e_{\slashed\nu} =e_\nu \in H^d_{P'}(\Gamma,\mathbb{Z}_o)$, which would bypass the spectral sequence computations in Appendix \ref{sec:lhs}.

Furthermore, if we consider the lift $B\mu^*$ as some form of "division by $q$" [\onlinecite{WPVZ}], then the filling factor obtained by this reconstruction can mimic those obtained above from the LSM anomaly.

\begin{remark}
\label{rmk:latred}
Recall the Bockstein map $\mathcal{B}:H^\ast_{P'}(\bullet,\mathbb{Z}/2)\rightarrow H^{\ast+1}_{P'}(\bullet,\mathbb{Z})$ sends the ($P'$-equivariant) second Stiefel-Whitney class $w^2_{P'}$ to the Euler class $e(P')$ in 3D. If the above can be done, then we would expect some class $B\mu^* w^2_{P'} = w_\nu^2$ to satisfy $e_\nu = \mathcal{B}w_\nu^2$. This class $w_\nu^2$ would classify the existence of $\text{Spin}_c$-structures on the unit cell $BP$ as a nilorbifold, justifying the observation [\onlinecite{WPVZ}] that the filling constraint can also be derived from the existence of $\text{Spin}_c$-structures on $\Lambda/P$. 
\end{remark}

The problem is that, with just knowledge of the point group $P'$ and the value of $\nu$, we cannot uniquely reconstruct the LSM anomaly. For instance, suppose we are given $P'=C_2\times C_6$ with $P$ characterized by $A=3$, one would then produce
\begin{equation}
3\nu_\omega \mod 2\mathbb{Z}. \nonumber
\end{equation}
However, we have no knowledge if the non-symmorphic element is represented unitarily ($3_2$-screw) or antiunitarily (glide-screw), hence lattice reduction does not tell the full story.  
%\begin{conjecture}
%The lattice reduction $\mu:\Lambda\rightarrow\Lambda_\nu$ allows one to reconstruct the lattice class $e_\nu =B\mu^*e_{\slashed\nu} \in H^d(P,\mathbb{Z}_o)$ only if the {\bf Lyndon-Serre-Hochschild spectral sequence}
%\begin{equation}
%E^{p,q}_2 = H^p(P',H^q(\Lambda,\mathbb{Z})) \Rightarrow H^n(P,\mathbb{Z}) \nonumber
%\end{equation}
%collapses at page three, with non-trivial entries $E_3^{0,q}\neq 0$ concentrated only on the $p=0$-th column.
%\end{conjecture}

\subsection{LSM Anomaly of the Quantum Hall System}\label{sec:lsmfqhe}
Let us now couple the QI to an external magnetic field. In this Quantum Hall (QH) scenario, a background, global electromagnetic $U(1)_E$ symmetry acts on all charged particles by a global phase rotation, and as such one may consider $U(1)_E \subset U(1)^N_C$ as the diagonal embedding. We first examine how $K_f$ arises from this $U(1)_E$-symmetry, then derive a cohomological description of fluxon-anyon braiding. We shall consider only translational symmetry $\mathbb{Z}^d\subset P$ for now.

\subsubsection{Fluxons and the magnetic domain}
The $U(1)_E$-curvature $b \in Z^2(\mathbb{Z}^d,U(1)_E)$ defines a magnetic translation algebra in analogy with Eq. (\ref{eq:obst}), but there exists a {\bf magnetic domain} $V_M\subset\Lambda_0$ on which the sublattice $\mathbb{Z}[V_M]$ is liftable into $\mathbb{Z}_\text{flux}^d$ [\onlinecite{LRO}]. Writing $b=e^{i2\pi\phi} t_1\cup t_2$ in terms of its $\mathbb{R}/\mathbb{Z}$-valued flux $\phi$, normalized in units of the flux quantum $\Phi_0=h/e$, then $\mathbb{Z}[V_M]$ is the minimal sublattice on which the flatness condition $\phi_{V_M}=0\mod \mathbb{Z}$ holds.

The fundamental domain/unit cell in the translationally-symmetric microscopic lattice $\Lambda_0$ is a face $f$ spanned by the primitive lattice vectors $\{t_\alpha\}_{\alpha\leq d}$. Let $L,m\in\mathbb{Z}$ denote the number of such faces in $\Lambda_0,V_M$, respectively.
\begin{definition}
The {\bf flux filling} $\nu_\phi$ is defined as the average number $\frac{1}{L}\sum\limits_{f\subset\Lambda_0}\phi_f =  \nu_\phi\Phi_0$ of flux quanta per fundamental domain. By definition of $V_M$, we have the quantization $m\nu_\phi \in \mathbb{Z}$.
\end{definition}
\noindent In a QH system with $N=1$ (the electron), the following anomaly-free filling constraint [\onlinecite{BMNS}], [\onlinecite{LRO}] has been derived:
\begin{equation}
\nu_\omega= \sigma_\text{Hall}\nu_\phi\mod\mathbb{Z},\qquad\sigma_\text{Hall}\in\mathbb{Z}, \label{eq:fluxhall}
\end{equation}
where $\sigma_\text{Hall}$ is the Hall conductance [\onlinecite{AS}], which always accompanies a flux-insertion by Laughling's flux-threading argument [\onlinecite{AS}], [\onlinecite{LRO}], [\onlinecite{SR}]. 

We may view Eq. (\ref{eq:fluxhall}) as a combinatorial "balance" between the particle- and the flux-filling factors, {\it given} that the Hall conductance $\sigma_\text{Hall}$ is quantized. Conversely, if we force $\nu_\phi\in\mathbb{Z}$ (by e.g. tuning the external magnetic field), then Eq. (\ref{eq:fluxhall}) implies the quantization of the Hall conductance $\sigma_\text{Hall}\in\mathbb{Z}$ at integral electron filling $\nu_\omega\in\mathbb{Z}$. In 2D, this is the integer QH (IQH) effect.

Now the integrality of the flux filling $\nu_\phi$ rests on the {\it commensurability condition} $m|L$ --- namely $V_M$ tiles the full microscopic lattice --- whence we may coarse-grain out $V_M$ such that $\phi=0\mod\mathbb{Z}$ in the thermodynamic limit. Otherwise, the flux is fractionalized and so must the ground state; this is the fractional Quantum Hall (FQH) effect [\onlinecite{LRO}], [\onlinecite{BBRFr}]. 

To capture such flux-induced fractionalization topologically, we generalize an argument in the literature [\onlinecite{LRO}], [\onlinecite{SR}]. 
\begin{proposition}
\label{prop:fluxtriv}
Given a magnetic domain $V_M$ of size $m$, the flux can be fractionalized without introducing a non-trivial GSD if $m'|k'L$, where $k=\operatorname{gcd}(k,m)k'$ and $m=\operatorname{gcd}(k,m)m'$.
\end{proposition}
\begin{proof}
Suppose first, on the lattice $\Lambda$, that the global electromagnetic $U(1)_E$ symmetry remains as background. The magnetic lattice $\Lambda_\text{flux}$ fits into the central extension sequence
\begin{equation} 
1\rightarrow U(1)_E\rightarrow \Lambda_\text{flux} \rightarrow\Lambda\cong\mathbb{Z}^d\rightarrow 1\nonumber
\end{equation}
characterized by the class $b\in H^2(\mathbb{Z}^d,U(1)_E)$. We also write a cocycle representative as $b=e^{i2\pi \phi}$ for some $\mathbb{R}/\mathbb{Z}$-valued flux $\phi\in Z^2(\mathbb{Z}^d,\mathbb{R}/\mathbb{Z})$. Now the coupling of this background $U(1)_E$ to the quasiparticles on the lattice $\Lambda$ gives rise to a {\it large gauge transformation} [\onlinecite{LRO}], [\onlinecite{SR}]
\begin{equation}
U_{f} = \exp\left(i 2\pi\phi_f\frac{1}{L} \sum_{e\in \partial f} P(e)\right),\nonumber
\end{equation}
where $\phi_{f}\in\mathbb{R}/\mathbb{Z}$ is the flux on the face $f\subset \Lambda$, and 
\begin{equation}
    P(e)=\frac{1}{k}\sum_{a\leq k}\int_edxx n_a(x)\nonumber
\end{equation}
is the average charge polarization. The flux is in fact uniform $\phi=\phi_f$ due to translational symmetry. 

By definition, $\mathbb{Z}[V_M]$ lifts to $\Lambda_\text{flux}$ and the extension cocycle $b(V_M)=0$ vanishes. The quantization $m\phi\in \mathbb{Z}$ then follows from the condition $U_{V_M}=1$. Now suppose $\frac{m}{k} = \frac{m'}{k'}$ for a coprime pair $\operatorname{gcd}(k',m')=1$, then its associated filling $\nu_\phi$ satisfies the fractionalization $\nu_\phi \in\frac{k'}{m'}\mathbb{Z}$.

In order to not introduce additional fractionalization on a closed system, however, $U_{\Lambda}=\prod\limits_{f\subset \Lambda}U_f=U_f^{L}$ must itself be a large gauge transformation. As only integral powers of $U_f$ are large gauge transformations themselves, this forces $L\nu_{\phi} \in \mathbb{Z}$, which is achieved only when $L\frac{k'}{m'} \in\mathbb{Z}$.
\end{proof}
\noindent Define $K_f=\mathbb{Z}/m$, maps $\xi\in\operatorname{Hom}(\mathbb{Z}/m,\mathbb{Z}/k)\xrightarrow{\text{eval}_1}\mathbb{Z}/\operatorname{gcd}(k,m)$ are determined by where it sends the generator $1\in\mathbb{Z}/m$. We now distill the above combinatorial argument to be independent of the lattice.

Consider the Diophantine equation $Lk + Km = S$, then B{\' e}zout's identity states that integer solutions $L,K\in\mathbb{Z}$ exist iff $S=0\mod\operatorname{gcd}(k,m)$. It is then possible to assign solutions $(L,K)$ of the inhomogeneous Diophantine equation 
\begin{equation}
Lk+Km=\xi(1)\mod \operatorname{gcd}(k,m),\label{eq:fluxfrac} 
\end{equation}
to $\xi(1)$. As the equation is unchanged under $(L,K)\mapsto (L+ak',K+am')$ for any $a\in\mathbb{Z}$, there is a series of group isomorphisms
\begin{equation}
\mathbb{Z}^2/[(k',m')\mathbb{Z}]\cong\mathbb{Z}/\operatorname{gcd}(k,m)\cong \operatorname{Hom}(\mathbb{Z}/m,\mathbb{Z}/k)\nonumber
\end{equation}
that uniquely associate classes of solutions to Eq. (\ref{eq:fluxfrac}) $[(L,M)]\mapsto \xi(1)\mapsto \xi$ to a map $\xi$. 

Next, the isomorphism $\mathbb{Z}/\operatorname{gcd}(k,m)\cong \frac{1}{\operatorname{gcd}(k,m)}\mathbb{Z}\subset \mathbb{Q}/\mathbb{Z}$ maps the Diophantine Eq. (\ref{eq:fluxfrac}) to the equation
\begin{equation}
k'L + m'K = \frac{\xi(1)}{\operatorname{gcd}(k,m)} \mod\mathbb{Z},\nonumber
\end{equation}
Recall that $\nu_\omega\in\frac{1}{k}\mathbb{Z}$ has $k$-torsion and $\nu_\phi\in\frac{1}{m} \mathbb{Z}$ has $m$-torsion modulo $\mathbb{Z}$. Given the quantization of the Hall conductance $\sigma_\text{Hall}\in\mathbb{Z}$, we define $(\nu_\phi,\nu_\omega)\in\mathbb{Q}^2$ by $\operatorname{lcm}[k,m](L,K)=(-\sigma_\text{Hall}m\nu_\phi,k\nu_\omega)$, then the class $[(L,K)]$ is isomorphic to flux-filling invariants characterized by the expression
\begin{equation}
-\sigma_\text{Hall}\nu_\phi+\nu_\omega - \theta_\xi \mod \mathbb{Z}, \label{eq:fluxfill}
\end{equation}
where the term $\theta_\xi\in \mathbb{Q}/\mathbb{Z}$ up to $\frac{1}{\operatorname{lcm}[k,m]}\mathbb{Z}$ is in one-to-one correspondence with the map $\xi\in\operatorname{Hom}(\mathbb{Z}/m,\mathbb{Z}/k)$. In particular $[(0,0)]\mapsto \theta_0=0\mod\mathbb{Z}$, in agreement with {\bf Proposition \ref{prop:fluxtriv}}.

%It is useful to describe the ways in which the fluxons $\phi$ are fused together across the faces $f\subset\Lambda_0$ with a series of "fluxon anomaly moves", as depicted in Fig. \ref{fig:fluxanom}.
%
%\begin{figure}[h]
%\centering
%\includegraphics[width=0.8\columnwidth]{fluxanom.png}
%\caption{The fluxon anomaly moves on a translationally-symmetric 2D lattice, satisfying bosonic fusion rules.}
%\label{fig:fluxanom}
%\end{figure}
%

\begin{remark}\label{rmk:quansur}
In an anomaly-free QH system with $\partial X\neq \emptyset$, a non-trivial boundary charge $\partial q$ is localized by a flux-threading in the bulk. If the gapped ground state remains unique, then gapless states must appear on $\partial X$ to carry this boundary charge. This is the {\bf many-body L{\" u}ttinger's theorem} [\onlinecite{LRO}], [\onlinecite{HB}].
\end{remark}

What the map $\xi\in\operatorname{Hom}(\mathbb{Z}/m,\mathbb{Z}/k)$ is depends on the underlying theory: it encodes, in the thermodynamic limit, the {\bf fluxon-anyon braiding} of the microscopic system [\onlinecite{LRO}]. If fractionalization of the GSD occurs due to Eq. (\ref{eq:fluxfill}), then $\varphi,\xi$ must contribute to the LSM anomaly. 

To see this, recall that the flux $\varphi$ is the $K_f$-valued coordinate functions of $\beta\in Z^2(\mathbb{Z}^d,K_f)$ in Eq. (\ref{eq:obst}). By the universal property of extension cocycles, there exists $\gamma\in C^1(\mathbb{Z}^d_\text{flux},K_f)$ such that $\beta=\delta\gamma$.
\begin{definition}
A {\bf splitting homomorphism} $\xi:K_f\rightarrow \mathbb{Z}/k$ of $\beta$ sends its coordinate function $\varphi \mapsto 1\mod k$ to the generator, where $k\in\mathbb{Z}$ is the torsion of a chosen class $\omega\in H^2(K_p,U(1))$.
\end{definition}
\noindent If a splitting map $\xi$ exists, then the LSM anomaly is given by $\xi(\gamma)\cup \omega \in H^3(K,U(1))$ [\onlinecite{ET}]. 

In the following, we shall generalize this notion by constructing a second order translationally-symmetric flux-threaded SPT phase $\mathcal{F}^{d+1}_{\mathbb{Z}^d}$ with the map $\xi$ as part of its data. Furthermore, we show that, under mild assumptions, its LSM anomaly is determined equivalently by a class of filling invariants $(\nu_\phi,\nu_\omega)$ characterized by Eq. (\ref{eq:fluxfill}).

\subsubsection{Fluxon-anyon braiding}
Recall the unitary modular tensor category (UMTC) $\mathcal{C}=\mathcal{C}^{d+1}_P$ defined by the lattice homotopy class $[\Lambda]\simeq \Omega \mathcal{C}$ corresponding to an order-$k$ element $\omega\in H^2(K_p,U(1))$. In the $(d+1)$D bulk, we prescribe the {\it extrinsic} braiding of the fluxon with the quasiparticles $[\Lambda]\simeq \Omega \mathcal{C}$ via an action $K_f\cong\mathbb{Z}/m \rightarrow \operatorname{Aut}\mathcal{C}$. 

By grading $\mathcal{C}_{\mathbb{Z}/m}=\bigoplus\limits_{l\in \mathbb{Z}/m}\mathcal{C}_l$ (with $\mathcal{C}_0 = \mathcal{C}$), the flux-threaded quantum magnet on $\Lambda$ is then tentatively modeled in the mesoscopic by the $\mathbb{Z}/m$-crossed UMTC $\mathcal{C}_{\mathbb{Z}/m}^\times$ [\onlinecite{BBCW}], [\onlinecite{Cui}].

It was shown [\onlinecite{BBCW}] that a non-trivial class $G\in H^3(\mathbb{Z}/m,\mathcal{C})$ obstructs an on-site implementation of the fluxon $\mathbb{Z}/m$-symmetry; in case $G=0$, the fractionalization class is classified by a projective class $\eta\in H^2(\mathbb{Z}/m,\mathcal{C})$. However,
\begin{proposition}\label{lem:gsdcent}
There is a one-to-one correspondence between $G\in H^3(\mathbb{Z}/m,\mathcal{C})$ and the maps $\xi\in\operatorname{Hom}(\mathbb{Z}/m,\mathbb{Z}/k)$.
\end{proposition}
\begin{proof}
Treat $\mathcal{C}$ as an Abelian UMTC with the set $[\mathcal{C}]$ of isomorphism classes of simple objects isomorphic to $\mathbb{Z}/k$, we may then specify the action $\mathbb{Z}/m \rightarrow \operatorname{Aut}\mathcal{C}\cong (\mathbb{Z}/k)^\times$ by its image $s\in (\mathbb{Z}/k)^\times$ of the generator $1\in\mathbb{Z}/m$. Up to isomorphism, we then have $H^3(\mathbb{Z}/m,[\mathcal{C}]) \cong H^3(\mathbb{Z}/m,\mathbb{Z}/k)$, in which elements $l\in\mathbb{Z}/m$ in the fluxon group act on $\mathbb{Z}/k$ via a multiplication by $ls$ --- a certain permutation of the excitations in $\mathcal{C}$.

Next, torsion cohomology of cyclic groups have period-two, with the degree-two isomorphism
\begin{equation}
    \cdot \cup\eta:H^{\ast}(\mathbb{Z}/m,\mathbb{Z}/k)\xrightarrow{\sim}H^{\ast+2}(\mathbb{Z}/m,\mathbb{Z}/k) \nonumber
\end{equation}
given by product against the generator $\eta\in H^2(\mathbb{Z}/m,\mathbb{Z}/k)$ [\onlinecite{Dav}], [\onlinecite{Hill}]. In particular, the odd-degree groups are isomorphic to the submodule $\operatorname{Tor}(\mathbb{Z}/k,\mathbb{Z}/(sm))$ with $sm$-torsion. As $s\in (\mathbb{Z}/k)^\times$ is invertible, this is equivalent to $\operatorname{Tor}(\mathbb{Z}/k,\mathbb{Z}/(sm))\cong \operatorname{Tor}(\mathbb{Z}/k,\mathbb{Z}/m)$. 

Invertibility of $s\in(\mathbb{Z}/k)^\times$ also allows us to identify $\xi\in \operatorname{Hom}(\mathbb{Z}/m,\mathbb{Z}/k)\cong H^1(\mathbb{Z}/m,\mathbb{Z}/k)$, hence as long as $\xi\neq 0$ we see that $G=\xi\cup\eta\neq 0$ is non-trivial.
\end{proof}
\noindent If $\xi\neq 0$, we cannot describe phases of flux-threaded lattice homotopy classes with $\mathcal{C}_{\mathbb{Z}/m}^\times$.

\begin{lemma}
\label{lem:fluxclass}
Let the translationally-symmetric flux-threaded lattice homotopy class $[\Lambda] \simeq \Omega\mathcal{F}_{\mathbb{Z}^d}^{d+1}$ determine a topological order $\mathcal{F}^{d+1}_{\mathbb{Z}^d}$.
\begin{enumerate}
\item It is classified by $H^{d+2}(\mathbb{Z}^d_\text{flux}\rtimes K_p,U(1))$.
\item If the extension class $\beta\in H^2(\mathbb{Z}^d,K_f)$ is trivial, then the LSM anomaly lives in $H^{d-1}(\mathbb{Z}^d,H^3(K,U(1)))$.
\end{enumerate}
\end{lemma}
\noindent By {\it Remark \ref{rmk:gspt}}, the order $\mathcal{F}^{d+1}_{\mathbb{Z}^d}$ defines uniquely an anomaly-free SPT phase $\mathcal{G}^{d+2}_{\mathbb{Z}^d}\cong\operatorname{Bulk}\mathcal{F}^{d+1}_{\mathbb{Z}^d}$, again by the holographic bulk-boundary correspondence [\onlinecite{LiWe}].
\begin{proof}
Recall that invertible $G$-protected SPT phases $\mathcal{F}^{d+1}_{\mathbb{Z}^d}$ in $(d+1)$D are classified by $H^{d+2}(G,U(1))$. Here the full symmetry group $G$ fits in the extension sequence
\begin{equation}
\begin{tikzcd}
1 \arrow[r] & K \arrow[r]                             & G \arrow[r]                                          & \mathbb{Z}^d \arrow[r] & 1 \\
1 \arrow[r] & K_f \arrow[r] \arrow[u, hook, "\iota"] & \mathbb{Z}^d_\text{flux} \arrow[r] \arrow[u, dashed] & \mathbb{Z}^d \arrow[r] \arrow[u, Rightarrow,no head]                                & 1
\end{tikzcd}\nonumber
\end{equation}
The inclusion on the left splits $K\cong K_f\rtimes K_p$, and hence $G \cong \mathbb{Z}^d_\text{flux}\rtimes K_p$. Now the universal class $\gamma\in H^1(\mathbb{Z}^d_\text{flux},K_f)$ has the property that $\iota^*\gamma$ is the tautological 1-cocycle generating $H^1(K_f,K_f)\subset H^1(K_f,U(1))$ [\onlinecite{ET}]. As such there is assigned on each lattice edge a $K_f$-valued defect $\iota^*\gamma$ associated to a magnetic translation, which we will interpret as an invertible (1+1)D defect network.

\begin{widetext}
\begin{center}
\begin{figure}[h]
\centering
\includegraphics[width=0.7\columnwidth]{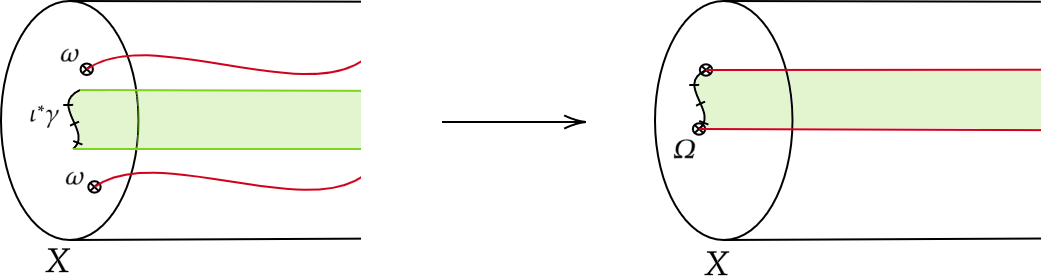}
\caption{The composite ((0+1)+1)D SPT phase obtained by combining the (0+1)D phase defined by a projective class $\omega\in H^2(K_p,U(1))$ and the (1+1)D phase defined by a magnetic translation $\iota^*\gamma \in H^1(K_f,U(1))$.}
\label{fig:composite}
\end{figure}
\end{center}
\end{widetext}

Now if $\omega\in H^2(K_p,U(1))$ is non-trivial, then we must describe a way to attach this (0+1)D phase to the boundary of the above (1+1)D phase; see Fig. \ref{fig:composite}. This is described by the action $F:K_p\rightarrow\operatorname{Aut}K_f$ built into the ((0+1)+1)D composite SPT phase. To classify such a phase, consider the Lyndon-Serre-Hochschild spectral sequence
\begin{equation}
E^{s,t}_2 = H^s(K_p,H^t(K_f,U(1)))\Rightarrow H^n(K,U(1)). \nonumber
\end{equation}
Entries in the degenerate page $E_r$ survive to the stable limit, and hence contribute to graded pieces $\operatorname{gr}_sH^n(G,U(1)) = E^{s,n-s}_r$ at degree $n$. The twisted product $\omega\otimes_F \iota^*\gamma \in H^2(K_p,H^1(K_f,U(1)))=E^{2,1}_2$ shows up at the $(s,t)=(2,1)$-th entry, which would contribute to $H^3(K,U(1))$ if degeneration occurs at page $r=2$. In general, the second order composite ((0+1)+1)D $K$-SPT phase is classified by an element $\Omega\in H^3(K,U(1))$ [\onlinecite{RL}].

We now stack these second-order phases according to translational symmetry. First, we assign $\Omega$ to an edge $\ell\in Z_1(\mathbb{Z}^d,\mathbb{Z})$ such that the 1-cocycle $\Omega_*\ell \in H_1(\mathbb{Z}^d,H^3(P,U(1)))$ is assembled via the (ungraded) coefficient map $$\mathbb{Z}\rightarrow \mathbb{Z}\cdot\Omega \in Z_1(\mathbb{Z}^d,\langle\Omega\rangle)\subset Z_1(\mathbb{Z}^d,H^3(K,U(1))).$$ We then perform homotopy descent for $(d-1)$ number of steps, which yields a class $\Omega_* T_{d-1} \in Z_{d-1}(\mathbb{Z}^d,H^3(K,U(1)))$ that lies in the image of the $(d-1)$-cycle $T_{d-1}\in Z_{d-1}(\mathbb{Z}^d,\mathbb{Z})$ under the coefficient map.

Since $H_\ast(\mathbb{Z}^d,\mathbb{Z})$ are free, we may once again use universal coefficient theorem $H_{d-1}(\mathbb{Z}^d,\mathbb{Z}) \cong H^{d-1}(\mathbb{Z}^d,\mathbb{Z})$, such that $\Omega_*T_{d-1}$ uniquely determines an element $\Omega_*T^{d-1} \in H^{d-1}(\mathbb{Z}^d,H^3(K,U(1)))$ classifying the topological order $\mathcal{F}^{d+1}_{\mathbb{Z}^d}$.

Next, we must identify $H^{d-1}(\mathbb{Z}^d,H^3(K,U(1)))$ as a subgroup of $H^{d+2}(G,U(1))$, which requires a K{\" u}nneth formula. Though the magnetic lattice $\mathbb{Z}_\text{flux}^d$ is not a split extension, we may use the universal class $\gamma$ to trivialize the extension cocycle $\beta=\delta \gamma\in Z^2(\mathbb{Z}^d,K_f)$, such that $\Omega_* T^{d-1}$ becomes exact in the Lyndon-Serre-Hochschild spectral sequence for $H^{d+2}(G,U(1))$ [\onlinecite{ET}].
\end{proof}
\noindent The above spectral sequence for $H^3(K,U(1))$ degenerates at page two only for $F\neq 0$ satisfying certain compatibility conditions [\onlinecite{AP}], [\onlinecite{DP}].

On the other hand, if we can find a section $\xi:K_f\rightarrow \langle\omega\rangle\cong\mathbb{Z}/k$ to the induced action $\operatorname{eval}_1F(\cdot):\langle\omega\rangle\cong\mathbb{Z}/k\rightarrow K_f$ evaluated at the generator $1\in K_f\cong\mathbb{Z}/m$, then it is possible to split the twisted product and kill the differential at the entry $\Omega=\omega\otimes_F\iota^*\gamma$ [\onlinecite{ET}].

\subsubsection{Flux-filling invariants}
In order to have some more control over the class $\Omega$, we impose the constraint 
\begin{eqnarray}
\delta\xi(\omega\otimes_F\gamma) &=& (\delta\omega)\cup\xi(\gamma) + \omega\cup\xi(\delta\gamma)\nonumber \\ &=& \omega\cup t_1\cup t_2
\xrightarrow{\iota^*}\omega\cup1=\delta\xi(\Omega).\nonumber
\end{eqnarray}
This condition means that the projective class $\omega\in H^2(K,U(1))$ lives at the endpoints of the 1D composite phase $\Omega$. The simplest solution to this condition is
\begin{equation}
\Omega=\Omega_\xi= [\xi(\varphi)^{-1}\cdot \omega] \cup \iota^*\gamma\in H^3(K,U(1))\nonumber
\end{equation}
where the inverse $\xi(\varphi)^{-1}$ is taken in the ring $\mathbb{Z}/k$; it is clear that splitting maps $\xi(\varphi)=1$ leave the projective class $\omega$ unmodified. Given the ($d$-component) generating cocycle $T^{d-1}\in H^{d-1}(\mathbb{Z}^d,U(1))$, the phase $\mathcal{F}^{d+1}_{\mathbb{Z}^d}$ is then classified by the crystalline LSM anomaly $(\Omega_\xi)_*T^{d-1}$. 

\begin{theorem}
\label{thm:fluxfrac}
The $K$-SPT phase $\mathcal{F}^{d+1}_{\mathbb{Z}^d}$ classified by $(\Omega_\xi)_*T^{d-1}\in H^{d-1}(\mathbb{Z}^d,H^3(K,U(1)))$ is equivalently classified by the flux-filling invariants $(\nu_\phi,\nu_\omega)$ of Eq. (\ref{eq:fluxfill}), given that the Hall conductance $\sigma_\text{Hall}\in\mathbb{Z}$ is quantized.
\end{theorem}
\begin{proof}
In the translationally-invaraint case, the filling invariants $(\nu_\phi,\nu_\omega)$ are simply identified as the roots of unity that are assigned to edges (and their endpoints) of $\Lambda$ by the second order phase $\Omega_\xi$ [\onlinecite{RL}]. 

Furthermore, translational symemtry also allows us to WLOG normalize the flux $\varphi=1\in\mathbb{Z}/m$ on each face/fundamental domain to the generator, whence we may associate the factor $\xi(\varphi)^{-1}=\xi(1)$ with a class $[(L,K)]$ of solutions to the Diophantine Eq. (\ref{eq:fluxfrac}). This in turn defines a class $(\nu_\phi,\nu_\omega)$ of flux-filling invariants characterized by Eq. (\ref{eq:fluxfill}). 

Conversely suppose we are given the indices $\nu_\phi,\nu$ with $m,k$-torsion modulo $\mathbb{Z}$, respectively, and the equation
\begin{equation}
-\sigma_\text{Hall}\nu_\phi + \nu_\omega - \theta \mod\mathbb{Z}.\nonumber
\end{equation}
If $\sigma_\text{Hall}\in\mathbb{Z}$, then the denominator of the fraction $\theta\in\mathbb{Q}/\mathbb{Z}$ is determined at most up to a multiple of $\operatorname{lcm}[k,m]$. 

Denote $[\theta]\in \frac{1}{\operatorname{lcm}[k,m]}\mathbb{Z}\cong\mathbb{Z}/\operatorname{lcm}[k,m]$ this particular fractional part, the above equation is then equivalent to the Diophantine Eq. (\ref{eq:fluxfrac}), whose inhomogeneous part $\xi_{[\theta]}=\xi(1)\in\mathbb{Z}/\frac{km}{\operatorname{lcm}[k,m]} = \mathbb{Z}/\operatorname{gcd}(k,m)$ determines a section $\xi:K_f\rightarrow \mathbb{Z}/k$.
\end{proof}
\noindent In case $\xi=0$, the class $\Omega_0$ yields the filling invariant $-\sigma_\text{Hall}\nu_\phi +\nu_\omega \mod\mathbb{Z}$, whose anomaly-free condition Eq. (\ref{eq:fluxhall}) describes the IQH effect. If $\xi\neq0$, we reproduce the QH filling constraint formula [\onlinecite{BBRFr}], [\onlinecite{LRO}]
\begin{equation}
-\sigma_\text{Hall}\nu_\phi + \nu_\omega = \theta_\xi \mod\mathbb{Z}\nonumber
\end{equation}
which describes the FQH effect. 

More precisely, it was demonstrated [\onlinecite{LRO}] that the minimal torsion $s$ modulo $\mathbb{Z}$ of $\theta_\xi\in\mathbb{Q}/\mathbb{Z}$ counts the statistics of fluxon-anyon braiding, and hence enhances the GSD. The anomaly-free condition then accounts for the quantization of $\sigma_\text{Hall}$ at {\it fractional} $\nu_\omega\in\frac{1}{s}\mathbb{Z}$ when $\nu_\phi\in\mathbb{Z}$ is integral. If $\xi=1$ is splitting, then the QH system has maximal GSD counted by $\operatorname{lcm}[k,m]$: every possible combination of the $k$-anyons and the $m$-fluxons leads to a disctinct fusion quasiparticle. 

This perspective corroborates with the composite fermion picture, at least combinatorially: the odd $s=3$ FQH state, in particular, arises from binding three fluxons to a single electron within a unit cell, which leads to the fractionalization $\nu_\omega\in\frac{1}{3}\mathbb{Z}$ as the size of the cyclotron orbit [\onlinecite{Jac}] is tripled by the three bound fluxons.

Our next goal is to lift {\bf Theorem \ref{thm:fluxfrac}} to the case with a possibly non-symmorphic space group $P$, and similarly classify the topological order $\mathcal{F}^{d+2}_P$ via $H^{d-1}(P,H^3(K,U(1)))\subset H^{d+2}(G,U(1))$. This is necessary for reproducing the filling constraints that have been derived previously [\onlinecite{BBRFr}], [\onlinecite{WPVZ}], [\onlinecite{LRO}].

\section{Non-Symmorphic Magnetic Algebra}\label{sec:nsymmag}
Recall the obstruction cocycle $\beta\in Z^2(\mathbb{Z}^d,K_f)$ defines the magnetic translations $\mathbb{Z}_\text{flux}^d$ by
\begin{eqnarray}
\beta(t,s) &=& \delta m(t,s)=[\tilde{t},\tilde{s}]\nonumber\\
&=&M(t+s)-M(s)-M(t)\in K_f,\nonumber
\end{eqnarray}
where $t,s\in\mathbb{Z}^d$ and $M=\tilde{\cdot}:\mathbb{Z}^d\rightarrow \mathbb{Z}_\text{flux}^d$ is a set-theoretic section; see Eq. (\ref{eq:obst}). If $P\cong \mathbb{Z}^d\rtimes P'$ is symmorphic, then by Bieberbach's theorem we may extend the domain of $b$ to $P$ by defining [\onlinecite{SSG}], [\onlinecite{GT}], [\onlinecite{GT1}]
\begin{eqnarray}
\beta(p,p')& =&\delta M(p,p')= [\tilde{p},\tilde{p'}] \nonumber\\
&=&M{(t+g\cdot t',gg')}M{(-g'^{-1}\cdot t',g'^{-1})}\nonumber \\
&\quad& \times M{(-g^{-1}\cdot t,g^{-1})} \nonumber
\end{eqnarray}
with $p = (t,g)\in \mathbb{Z}^d\rtimes P'$, such that $b\in Z^2(\mathbb{Z}^d\rtimes P',K_f)$ is again a 2-cocycle.

When $P$ is non-symmorphic, the extension class $\nu\in H^2(P',\mathbb{Z}^d)$ obstructs the lifting $S=\hat{\cdot}:P'\rightarrow P$; see Eq. (\ref{eq:obst}). As such the quantity
\begin{equation}
[\tilde{p},\tilde{p}']=\widetilde{pp'}\tilde{p'}^{-1}\tilde{p}^{-1},\qquad p,p'\in P' \nonumber
\end{equation}
does not admit a well-defined extension into $P$. The ambiguity here is in the first factor, where the product $pp'$ can either lift to $\widehat{pp'}$ or $\hat{p}\hat{p}'$. These quantities differ by a central element $\nu(p,p')\in\mathbb{Z}^d$, and we may supplement this in the above by defining
\begin{equation}
\beta_\nu(p,p') = \widetilde{\nu(p,p')}\cdot\widetilde{S(p)S(p')}\widetilde{S(p')}^{-1}\widetilde{S(p)}^{-1}.\nonumber
\end{equation}
This casts $\beta_\nu\in Z^{2;\nu}(P',K_f)$ as a {\it twisted} 2-cocycle. Given the choice of sections $M,S$, it forms a model for cocycles on $P$ --- the class $\beta_\nu\in H^{2;\nu}(P',K_f)$ does not depend on the choice of these sections [\onlinecite{GT}], [\onlinecite{GT1}].

Now similar to how $\beta\in H^2(\mathbb{Z}^d,K_f)$ classifies the magnetic translations $\mathbb{Z}_\text{flux}^d$, the twisted version $\beta_\nu\in H^{2;\nu}(P',K_f)$ classifies the {\bf non-symmorphic magnetic space group} $\mathcal{P}$ sitting in the extension sequence 
\begin{equation}
1\rightarrow K_f\rightarrow \mathcal{P}\rightarrow P\rightarrow 1. \nonumber
\end{equation}
We may define the {\bf area cocycle} $c_\nu \in Z^{2;\nu}(P',\mathbb{R})$ [\onlinecite{MM}], [\onlinecite{MT}], which computes the area of the geodesic polygon spanned by $p,p'\in P$, such that
\begin{equation}
\beta_\nu = \varphi \cdot c_\nu \in H^{2;\nu}(P',K_f). \label{eq:nsymmagalg}
\end{equation}
The presence of the non-symmorphic twist $\nu$ implies that, if the fundamental domain $\Gamma$ is normalized to have integral area, $c_\nu$ can be $\mathbb{Q}$-valued. As such, the flux $\varphi$ can be enhanced purely by nature of the lattice.

We once again trivialize $\beta_\nu=\delta\gamma_\nu$ for a cocycle $\gamma_\nu\in Z^1(\mathcal{P},K_f)$, and suppose we can find a section $\xi:K_f\rightarrow\mathbb{Z}/k$ that splits the twisted product $\omega\otimes_F \gamma_\nu$. The corresponding condition $\delta\xi(\omega\otimes_F \gamma_\nu) = \omega\cup c_\nu$ then admits a solution
\begin{equation}
\tilde\Omega_\xi =[\xi(\varphi)^{-1}\cdot\omega]\cup \iota^*\gamma_\nu\in H^{3;\nu}(K,U(1))\label{eq:nsmagcl}
\end{equation}
as the building blocks for the crystalline LSM anomaly; the condition for $\xi$ to be splitting becomes $\delta \xi(\gamma_\nu)=c_\nu$. The subtle point here is that $\varphi$ is the flux associated to the {\it unit cell} $BP$, which can have fractional area, so we cannot directly copy the proof of {\bf Theorem \ref{thm:fluxfrac}}.

Now in order to compute the crystalline LSM anomaly, we make use of codimension-1 hyperplane cycles $h_\nu \in Z_{d-1}(P,\mathbb{Z}_o)$ to perform the homotopy descent of the class $\tilde\Omega_\xi$ Eq. (\ref{eq:nsmagcl}), such that the topological order $\mathcal{F}^{d+1}_P$ is described by
\begin{equation}
(\tilde{\Omega}_\xi)_*h_\nu\in H_{d-1}(P,H^{3;\nu}(K,U(1))) \nonumber
\end{equation}
in analogy with {\bf Lemma \ref{lem:fluxclass}}. The "Poincar{\' e} dual" of this description, $H^1(P,H^{3;\nu}(K,U(1)))$ [\onlinecite{RL}], prescribes an explicit procedure for constructing $\mathcal{F}^{d+1}_P$ via {\it non-symmorphic} decorated domain walls.

\subsection{Non-Symmorphic Decorated Domain Walls}
Recall that a first order (0+1)D $K$-SPT phase is classified by a projective class $\omega\in H^2(K,U(1))$. In the decorated domain wall construction, each lattice site in $\Lambda$ is treated as a 0D domain wall decorated by $\omega$, then stacked according to the space group $P$. The non-symmorphic effects are then captured by the lattice class $e_\nu \in H^d(P,\mathbb{Z}_o)$ in Appendix \ref{sec:lhs}. 

For the second order composite ((0+1)+1)D phase classified by $\tilde\Omega_\xi\in H^3(K,U(1))$, we instead consider codimension-1 decorated domain walls [\onlinecite{RL}] on $\Lambda$. If $P\cong \mathbb{Z}^d\rtimes P'$ is symmorphic, and $\Gamma$ is the fundamental domain, then the class $\Omega_\xi \in H^3(K,U(1))$ decorates intersections $h_\nu\cap \Gamma^{P'}$ of domain walls $h_\nu$ and the fixed-point set $\Gamma^{P'}$ with linear (unitary) representations $H^1(P',\langle\Omega_\xi\rangle)$ of $P'$, in order to ensure $\mathcal{F}^{2+1}_{P'}$ defines a gapped phase on $\Gamma$ [\onlinecite{RL}]. 

If $P$ is non-symmorphic, however, we cannot make use of the fundamental domain. In the following, we shall leverage our knowledge of the extension class $\nu\in H^2(P',\Lambda)$ and compute twisted representations $H^{1;\nu}(P',\langle\tilde\Omega_\xi\rangle)$ as a model for $H^1(P,\langle\tilde\Omega_\xi\rangle)$. This describes the {\bf non-symmorphic decorated domain wall} construction that constitutes the phase $\mathcal{F}^{d+1}_P$.
\medskip

\subsubsection{(2+1)D Glide-magnetic order $\mathcal{F}^{2+1}_{\mathtt{pg}}$}
Consider first the 2D glide lattice $P=\mathtt{pg}$ with the non-trivial extension class $\nu\in H^2(C_2,\Lambda)\cong\mathbb{Z}/2$, and let $\Lambda_1=\Lambda_1^{C_2}\subset\Lambda$ denote the glide axis fixed under the point group $P'=C_2$. If $\ell\cap \Lambda_1\neq \emptyset$, then the class $(\tilde\Omega_\xi)_*\ell$ assigned to this edge forms a twisted representation in $H^{1;\nu}(C_2,\langle\tilde\Omega_\xi\rangle)$ of $C_2$. 

Through the section $\hat{\cdot}=S:C_2\rightarrow\mathtt{pg}$, such twisted representations are lifted from an ordinary one $W\in H^1( C_2,\langle\tilde\Omega_\xi\rangle)$ via the condition $W_{\hat{g}^2}=W_{\hat{g}}^2=W_{t_1}W_{\widehat{g^2}}=W_{t_1}$ is satisfied, where $g\in C_2$ is the reflection generator and $t_1\in\Lambda$ generates the glide axis $\Lambda_1=\mathbb{Z}[t_1]$. As such, a glide reflection $W_{\hat{g}}$ carries half the index carried by a translation $W_{t_1}$; in other words, a twisted representation in $H^{1;\nu}(C_2,\langle\tilde\Omega_\xi\rangle)$ is determined by a $C_2=\mathbb{Z}/2$-extension $\operatorname{Ext}(\langle t_1\rangle,\mathbb{Z}/2)$ of the subgroup $\langle t_1\rangle\subset H^1(\mathbb{Z}^2,\langle\tilde\Omega_\xi\rangle)$ generated by $t_1$.

Consider an edge $\ell^\perp$ intersecting $\Lambda_1$ only at a site $x_\ast\in\ell^\perp\cap \Lambda_1$. Recall $\tilde\Omega_\xi$ has the projective class $\omega$ at its endpoints, hence this 1D domain wall $\ell^\perp$ is in fact decorated with a $\langle\omega\rangle=\mathbb{Z}/k$-valued twisted representation. As $\langle t_1\rangle \cong\mathbb{Z}/k$, the glide reflection $W_{\hat{g}}$ is characterized by $\operatorname{Ext}(\mathbb{Z}/k,\mathbb{Z}/2) \cong \mathbb{Z}/\operatorname{gcd}(k,2)$; see Fig. \ref{fig:glidemag}. This suffices to determine the index of the glide operator $W_{\hat{g}}$, and hence the twisted representation $W\in H^{1;\nu}(C_2,\mathbb{Z}/k)$ assigned to $\ell^\perp$. Here we see that, if $k$ is odd, then the projective class $\omega$ cannot be split by the glide, and a weak glide-symmetric phase manifests.

\begin{widetext}
\begin{center}
\begin{figure}[h]
\centering
\includegraphics[width=0.5\columnwidth]{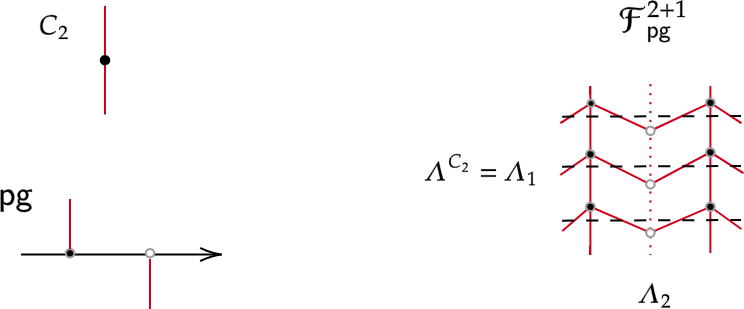}
\caption{The glide-magnetic topological order $\mathcal{F}^{2+1}_\mathtt{pg}$. {\bf Left}: The 1D domain walls (red) are assigned ordinary representations $H^{1}(C_2,\langle\tilde\Omega_\xi\rangle)$, while the lattice points (black dot) are assigned $H^{1}(C_2,\langle\omega\rangle)$. Due to glide symmetry, these representations are split in two (black and white dots with grey outline), characterized by a $\mathbb{Z}/2$-extension. {\bf Right}: The decorated 1D domain walls defining the glide-magnetic topological order $\mathcal{F}^{2+1}_\mathtt{pg}$.}
\label{fig:glidemag}
\end{figure}
\end{center}
\end{widetext}

Now consider an edge $\ell^{||}\subset \Lambda_1$. The structure of $\tilde\Omega_\xi$ allows us to split the extension classes \begin{eqnarray}
\operatorname{Ext}(\langle t_1\rangle,\mathbb{Z}/2) &\cong& \operatorname{Ext}(\mathbb{Z}/k, \mathbb{Z}/2)\oplus \operatorname{Ext}(\mathbb{Z}/m,\mathbb{Z}/2)\nonumber\\
&\cong&\mathbb{Z}/\operatorname{gcd}(k,2)\oplus \mathbb{Z}/\operatorname{gcd}(m,2),\nonumber
\end{eqnarray} the components of which determine a twisted representation in the "electric" $W_e\in H^{1;\nu}(C_2,\mathbb{Z}/k)$ and the "magnetic" sectors $W_m\in H^{1;\nu}(C_2,\mathbb{Z}/m)$, respectively. As the endpoints $x_\ast$ of $\ell^{||}$ coincide with those of $\ell^{\perp}$, the electric-sector decoration on $\ell^{||}$ --- as an element of $H^{1;\nu}(C_2,\mathbb{Z}/k)$ --- must coincide with that on $\ell^{\perp}$. This allows us to build the (2+1)D glide-magnetic topological order $\mathcal{F}^{2+1}_\mathtt{pg}$ as shown in Fig. \ref{fig:glidemag}.

To understand the fluxes $\varphi$ that live on the glide lattice, we compute $c_\nu$. From any lattice site $0\in \Lambda$, the points $0,\hat{g}\cdot 0,\hat{g}^2\cdot 0$ defined by the glide element $\hat{g}\in\mathtt{pg}$ span a convex triangle. If we normalize the quadrangle spanned by $t_1,t_2$ to unit area, then this triangle has area $c_\nu(\hat{g},\hat{g}) = \frac{1}{2}\cdot 1 = \frac{1}{2}$. To ensure $\beta_\nu$ remains $K_f$-valued, the flux $\varphi \in 2\cdot K_f$ must be doubled on each fundamental domain.

Now as there are no subgroups of $C_2$ that lift faithfully into $\mathtt{pg}$, the indices of the twisted representations $W_e,W_m$ are determined purely by the extension classes $f_k\in\operatorname{Ext}(\mathbb{Z}/k,\mathbb{Z}/2),f_m\in\operatorname{Ext}(\mathbb{Z}/m,\mathbb{Z}/2)$. If $\eta_k,\eta_m\in\frac{1}{2}\mathbb{Z}$ denote their 2-torsion indices, we have
\begin{eqnarray}
W_e &=& f_k=e^{i\frac{2\pi}{k} \eta_k}\in H^{1;\nu}(C_2,\mathbb{Z}/k),\nonumber\\
W_m &=& f_m=e^{i\frac{2\pi}{m}\eta_m}\in H^{1;\nu}(C_2,\mathbb{Z}/m). \nonumber
\end{eqnarray}
As mentioned previously, these indices $\eta_k,\eta_m$ fractionalize those $\nu_\omega\in \frac{1}{k}\mathbb{Z},\nu_\phi\in\frac{1}{m}\mathbb{Z}$ associated to translation symmetry $\mathbb{Z}^2$. If the extension classes $f_k,f_m$ are non-trivial, then $\nu_\omega,\nu_\phi$ must be doubled.

\subsubsection{(3+1)D Screw-magnetic order $\mathcal{F}^{3+1}_{P_{A;\nu}}$}
Let us now move on to the 3D case. Consider the screw lattice $\Lambda$ with the axial rotational point group $P'=1\times C_n$; recall the $A_{\frac{n}{A}}$-screw element in $P=P_{A;\nu}$ is characterized by a divisor $A$ of $n$ for which $\operatorname{gcd}(A,\frac{n}{A})=1$, and a choice of a non-trivial extension class $\nu\in H^2(P',\Lambda)\cong \mathbb{Z}/\frac{n}{A}$. Let $\Lambda_1 = \Lambda^{1\times C_n}$ denote the fixed point axis along the screw element. We adopt the same principle as above: each 2D domain wall $f\subset \Lambda$ is decorated with a twisted representations $H^{1;\nu}(P',\langle\tilde\Omega_\xi\rangle)$ determined by the intersection $f\cap \Lambda_1$.

To construct a twisted representation, we impose the condition $W_{\hat{g}^n}=W_{\hat{g}}^n= W_{|\nu| t_1}W_{\widehat{g^n}} = W_{|\nu| t_1}$ on an ordinary representation $W\in H^1(C_n,\langle\tilde\Omega_\xi\rangle)$, where $|\nu|$ is the order of $\nu$ in $\mathbb{Z}/\frac{n}{A}$. Now by definition, the subgroup $C_A\subset C_n$ lifts faithfully into $P_{A;\nu}$, hence the section $\hat{\cdot}=S:C_n\rightarrow P_{A;\nu}$ descends to a group homomorphism on $C_A$ for which $S^*W$ defines an {\it untwisted} representation. As such, the above condition reads $W_{[\hat{g}]}^{\frac{n}{A}} = W_{\nu t_1}$ on each coset $[g]\in C_n/C_A$, determining a $C_n/C_A=\mathbb{Z}/\frac{n}{A}$-extension of $\nu \cdot\langle t_1\rangle\subset H^1(\mathbb{Z}^3,\langle\tilde\Omega_\xi\rangle)$ characterized by $\operatorname{Ext}(\nu\cdot\langle t_1\rangle,\mathbb{Z}/\frac{n}{A})$. A screw rotation $W_{\hat{g}}$ must carry an index $\frac{A}{n}\in\mathbb{Q}$ times that carried by a $\nu$-fold translation.

For definiteness, take $\nu=1$ to be the generator. Let $h^\perp$ dentote a 2D hyperplane (the dotted blue triangles in Fig. \ref{fig:screwflux}) such that $h^\perp\cap\Lambda_1$ is a single point. Once again, as the projective class $\omega$ lies at the endpoints of $\tilde\Omega_\xi$, we assign to the 2D domain wall $h^\perp$ a $\langle\omega\rangle = \mathbb{Z}/k$-valued twisted representation on the intersection $x_\ast = h^\perp\cap \Lambda_1$. These are constructed from an ordinary representation $W\in H^1(C_A,\mathbb{Z}/k) \cong \mathbb{Z}/\operatorname{gcd}(A,k)$, then associating an extension class $\operatorname{Ext}(\mathbb{Z}/k,\mathbb{Z}/\frac{n}{A})\cong\mathbb{Z}/\operatorname{gcd}(k,\frac{n}{A})$ to the quotient $C_n/C_A=\mathbb{Z}/\frac{n}{A}$. Here we see that, if $\operatorname{gcd}(\frac{n}{A},k)=1$, then the linear $C_n$-representation $W\in \mathbb{Z}/\operatorname{gcd}(n,k)$ cannot be split by the screw.

\begin{widetext}
\begin{center}
\begin{figure}[h]
\centering
\includegraphics[width=0.5\columnwidth]{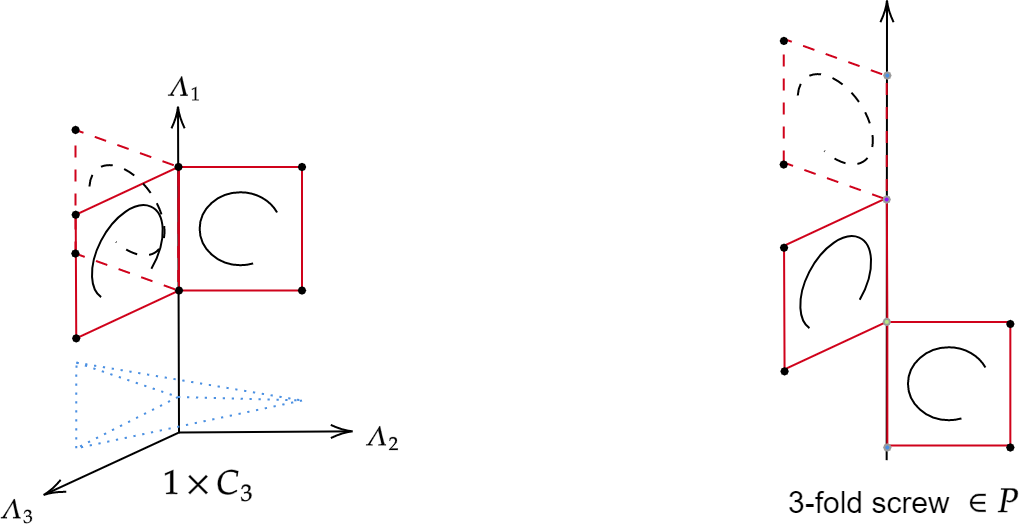}
\caption{The screw-magnetic topological order $\mathcal{F}^{3+1}_{P_{1;1}}$ of order-$3$, with $P'=1\times C_3$, $A=1$ and $\nu=1$. {\bf Left}: The 2D decorated domain walls $h^{||}$ in red are assigned regular representations $H^1(C_n,\langle\tilde\Omega_\xi\rangle)$, while those in blue $h^\perp$ are assigned $H^1(C_n,\langle\omega\rangle)$. {\bf Right}:  Due to a 3-fold screw, these representations are split (blue, green and purple dots with grey outline) according to an extension class in $\operatorname{Ext}(\langle t_1\rangle,\mathbb{Z}/3)$.}
\label{fig:screwflux}
\end{figure}
\end{center}
\end{widetext}

Next, let $h^{||}$ denote a 2D hyperplane colinear with the screw axis $\Lambda_1$, such that they intersect at a 1D edge $\ell= h^{||}\cap \Lambda_1$. We then assign a twisted representation $H^{1;\nu}(C_n,\langle\tilde\Omega_\xi\rangle)$ to it. By the structure of $\tilde\Omega_\xi$, we may split both the untwisted representations on $C_A$
\begin{eqnarray}
H^1(C_A,\langle\tilde\Omega_\xi\rangle)&\cong& H^1(C_A,\mathbb{Z}/k)\oplus H^1(C_A,\mathbb{Z}/m) \nonumber\\
&=& \mathbb{Z}/\operatorname{gcd}(A,k)\oplus\mathbb{Z}/\operatorname{gcd}(A,m), \nonumber
\end{eqnarray}
as well as the extension classes on the quotient $C_n/C_A$ 
\begin{eqnarray}
\operatorname{Ext}(\langle t_1\rangle,\mathbb{Z}/\frac{n}{A}) &\cong& \operatorname{Ext}(\mathbb{Z}/k,\mathbb{Z}/\frac{n}{A})\oplus \operatorname{Ext}(\mathbb{Z}/m,\mathbb{Z}/\frac{n}{A})\nonumber\\
&=&\mathbb{Z}/\operatorname{gcd}(k,\frac{n}{A})\oplus \mathbb{Z}/\operatorname{gcd}(m,\frac{n}{A}),\nonumber
\end{eqnarray} 
into electric and magnetic sectors, decorating the endpoints and interior of $\ell=h^{||}\cap\Lambda_1$, respectively. Furthermore, the data in the electric sector must match those assigned on $h^\perp$. 

Given screw-rotations $\hat{g},\hat{g}'$ for which $[g]\neq [g']=[g]+1\in C_n/C_A$ lie in distinct but adjascent cosets, the non-symmorphic area cocycle $c_\nu$ computes the area $c_\nu(\hat{g},\hat{g}')=\frac{A}{n}\cdot 1$ of the wedge spanned by $0,\hat{g}\cdot 0,\hat{g}'\cdot 0$, where we have normalized each $n$-gon perpendicular to $\Lambda_1$ to have unit area. We must then have $\varphi\in\frac{n}{A}\cdot K_f$ to ensure that $\beta_\nu$ is $K_f$-valued. This allows us to build the (3+1)D screw-magnetic order $\mathcal{F}^{3+1}_{P_{A;\nu}}$ as shown in Fig. \ref{fig:screwflux}.

The indices of the untwisted representations on $C_A$ are determined by
\begin{equation}
\mu \in \frac{k}{\operatorname{gcd}(A,k)}\mathbb{Z},\qquad \mu_\phi \in \frac{m}{\operatorname{gcd}(A,m)}\mathbb{Z}.\nonumber
\end{equation}
Pick a pair of extensions $f_k\in \operatorname{Ext}(\mathbb{Z}/\frac{n}{A},\mathbb{Z}/k),f_m\in\operatorname{Ext}(\mathbb{Z}/\frac{n}{A},\mathbb{Z}/m)$, and let $\eta_k,\eta_m$ denote their roots of unity. We lift the above indices to $P_{A;\nu}$ by
\begin{eqnarray}
W_e(\hat{g}) &=& f_k([g])\cdot \exp i\frac{2\pi}{k} \mu\nonumber\\
&\qquad&= e^{i\frac{2\pi}{k}(\eta_k([g]) + \mu)}\equiv e^{i2\pi \bar{\mu}(\eta_k)}\in\mathbb{Z}/k ,\nonumber \\ 
W_m(\hat{g}) &=& f_m([g])\cdot \exp i\frac{2\pi}{m} \mu_\phi\nonumber\\
&\qquad&=e^{i\frac{2\pi}{m}(\eta_m([g])+\mu_\phi)}\equiv e^{i2\pi \bar{\mu}_\phi(\eta_m)}\in\mathbb{Z}/m,\nonumber
\end{eqnarray}
where $g\in C_n$ is the generator; this suffices to determine $W_e\in H^{1;\nu}(P_{A;\nu},\mathbb{Z}/k)$ and $W_m\in H^{1;\nu}(P_{A;\nu},\mathbb{Z}/m)$. By definition of $f_k,f_m$, the indices $\eta_k,\eta_m$ have $\frac{n}{A}$-torsion $\frac{n}{A}\cdot \eta_k=0,\frac{n}{A}\cdot\eta_m=0$ modulo $\mathbb{Z}$.

\subsubsection{(3+1)D Glide-screw-magnetic order $\mathcal{F}_\mathtt{gs}^{3+1}$}
Recall that a glide-screw is an order-2 screw coupled to a glide reflection. The unique glide-screw lattice $\Lambda$ is characterized by the non-trivial extension class $\nu = (1,1)\in H^2(P',\Lambda) \cong\mathbb{Z}/2\oplus\mathbb{Z}/2$, where $P' = C_2\times C_6$ is the chiralrotational point group of order 12. The corresponding space group $P=\mathtt{gs}$ is determined by the divisor $A'=3$ of $n=6$ for which $\frac{n}{A'}=2$ is even. 

Similar to the above, each decorated 2D domain wall $f\subset \Lambda$ is assigned a twisted representation $H^{1;\nu}(C_2\times C_6,\langle\tilde\Omega_\xi\rangle)$ according to its intersection $f\cap \Lambda_1$ with the fixed point set $\Lambda_1=\Lambda^{C_2\times C_n}$, which we take to be the glide-screw axis. 

Starting from an ordinary representation $W\in H^1(C_2\times C_6,\langle\tilde\Omega_\xi\rangle)$, we impose the constraint $W_{\widehat{(r,[g])}}^2 = W_{t_1}W_{\widehat{(r,[g])^2}} = W_{t_1}$, where $(r,g)\in  C_2\times C_6$ is the generator and $[g] \in C_6/C_3= C_2$ is the non-trivial coset. This defines a double-extension characterized by $\operatorname{Ext}(\mathbb{Z}/2,\langle t_1\rangle)$, with $\mathbb{Z}/2 \subset C_2\times C_6/C_3$ the diagonal subgroup. Due to the structure of $\tilde\Omega_\xi$, we may once again split the representations into electric and magnetic sectors:
\begin{eqnarray}
W_e(\widehat{r,g}) &=& f_k(r,[g])\cdot \exp i\frac{2\pi}{k} \mu\nonumber\\
&\qquad&= e^{i\frac{2\pi}{k}(\eta_k([g]) + \mu)}\equiv e^{i\frac{2\pi}{k} \bar{\mu}(\eta_k)}\in\mathbb{Z}/k ,\nonumber \\ 
W_m(\widehat{r,g}) &=& f_m(r,[g])\cdot \exp i\frac{2\pi}{m} \mu_\phi\nonumber \\
&\qquad& =e^{i\frac{2\pi}{m}(\eta_m([g])+\mu)}\equiv e^{i\frac{2\pi}{m} \bar{\mu}_\phi(\eta_m)}\in\mathbb{Z}/m,\nonumber
\end{eqnarray}
from a pair of extensions $f_k\in \operatorname{Ext}(\mathbb{Z}/2,\mathbb{Z}/k),f_m\in\operatorname{Ext}(\mathbb{Z}/2,\mathbb{Z}/m)$ and their indices
\begin{equation}
\mu \in \frac{k}{\operatorname{gcd}(3,k)}\mathbb{Z},\qquad \mu_\phi\in\frac{m}{\operatorname{gcd}(3,m)}\mathbb{Z}.\nonumber
\end{equation}
The flux $\varphi$ can similarly be computed to lie in $2\cdot K_f$ by previous arguments.

However, here the glide-screw operator $W_{\widehat{(r,[g])}}=W_e(\widehat{r,g}) \cdot W_m(\widehat{r,g})$ must be antiunitary, hence the indices $\bar{\mu}(\eta)=\bar{\mu}(\eta)+2,\bar\mu_\phi(\eta) =\bar{\mu}_\phi(\eta)+2$ are only defined mod-2 [\onlinecite{RL}]. They are determined by the torsion group in the exact sequence
\begin{equation}
0 \rightarrow \operatorname{Tor}(\mathbb{Z}/2,H) \rightarrow H\xrightarrow{2\cdot} H\rightarrow \mathbb{Z}/2\otimes H\rightarrow 0,\nonumber
\end{equation}
where $H = H^1(C_3,\langle\tilde\Omega_\xi\rangle) \otimes \operatorname{Ext}(\langle t_1\rangle,\mathbb{Z}/2)$. Since $\operatorname{gcd}(2,3)=1$, we have expectedly 
\begin{eqnarray}
\operatorname{Tor}(\mathbb{Z}/2,H_e) = \mathbb{Z}/\operatorname{gcd}(k,2),\nonumber\\ \operatorname{Tor}(\mathbb{Z}/2,H_m) = \mathbb{Z}/\operatorname{gcd}(m,2)\nonumber
\end{eqnarray}
in each electric/magnetic sector $H=H_e\oplus H_m$. This mod-2 reduction means that, in general, distinct tuples of linear $C_3$-representations and extension classes can give rise to the same index $\bar\mu,\bar\mu_\phi$. 

\subsection{Non-Symmorphic Flux-Filling Invariants}
One may notice that the common theme in the above is the fractionalization of the indices $(\nu_\phi^1,\nu^1_\omega)$ associated to translational symmetry. These indices classify the representations $H^1(\mathbb{Z}^d,\langle\tilde\Omega_\xi\rangle)$, and are related to the flux-filling factors $(\nu_\phi,\nu_\omega)$ determined by the Diophantine equation Eq. (\ref{eq:fluxfrac}). More precisely, if we identify $(\nu_\phi,\nu_\omega)$ with the indices of a representation $U=T^{\cup d}\in H^d(\mathbb{Z}^d,\langle\tilde\Omega_\xi\rangle)$ for $T\in H^1(\mathbb{Z}^d,\langle\tilde\Omega_\xi\rangle)$, we may write
\begin{equation}
\nu_\omega = \frac{1}{d}\sum_{\alpha\leq d}\nu^\alpha_\omega,\qquad \nu_\phi= \frac{1}{d}\sum_{\alpha\leq d}\nu_\phi^\alpha,\nonumber
\end{equation}
where the indices $(\nu_\phi^\alpha,\nu^\alpha_\omega)$ classify the representation $T = (T_1,\dots,T_d) \in H^1(\mathbb{Z}^d,\langle\tilde\Omega_\xi\rangle)$ for $1\leq \alpha\leq d$. Now as the non-symmorphic effects of interest throughout the above is concentrated along only one axis $\alpha=1$, the order of the extension classes $f_k,f_m$ determines the fractionalization of $(\nu^1_\phi,\nu^1_\omega)$, which in turn determines that of the flux-filling invariant $(\nu_\phi,\nu_\omega)$.

\subsubsection{Glide-magnetic flux-filling invariants}
In the 2D glide case $P=\mathtt{pg}$, the extension classes $f_k\in\operatorname{Ext}(\mathbb{Z}/k,\mathbb{Z}/2),f_m\in\operatorname{Ext}(\mathbb{Z}/m,\mathbb{Z}/2)$ have order-2, and as such
\begin{equation}
\nu^1_\phi,\nu_\phi \in 2\cdot \frac{1}{m}\mathbb{Z},\qquad \nu^1_\omega,\nu_\omega \in 2\cdot \frac{1}{k}\mathbb{Z}.\nonumber
\end{equation}
The next ingredient required to assemble the flux-filling invariant is the factor $\xi(\varphi)^{-1}$. To compute this term, recall that the non-symmorphic area $c_\nu(\hat{g},\hat{g}) = \frac{1}{2}\cdot 1$ on the glide element $\hat{g}\in \mathtt{pg}$ is half that of each fundamenal domain. In order for $\beta_\nu$ to remain $K_f$-valued, we must take $\varphi=2\cdot \phi$ for some $\phi \in K_f$.

With this, we normalize the flux $\xi(\phi)^{-1}$ to $\xi(1)^{-1}$, whence the factor $\xi(\varphi)^{-1}$ determines $\theta_\xi \in \mathbb{Q}/\mathbb{Z}$ with denominator $\frac{km}{2\operatorname{gcd}(k,m)}=\frac{1}{2}\operatorname{lcm}[k,m]$. The class $(\tilde\Omega_\xi)_*h_\nu$ then leads to the flux-filling invariant
\begin{equation}
-\sigma_\text{Hall}\nu_\phi + \nu_\omega - \theta_\xi \mod 2\mathbb{Z}.\label{eq:glideflux}
\end{equation}
If $\xi$ were trivial, then the anomaly-free condition $\nu_\omega = \sigma_\text{Hall}\nu_\phi\mod 2\mathbb{Z}$ forces the doubling of both $\nu_\omega$ and $\nu_\phi$, consistent with the result that weak glide-symmetric phases square to the trivial phase [\onlinecite{XA}]. 

On the other hand, we must have $d\geq 2$ in order to observe effects of fluxon-anyon braiding $\xi\neq 0$, in which case the filling constraint on Eq. (\ref{eq:glideflux}) had been previously derived [\onlinecite{LRO}] through physical arguments. In particular, if the flux filling $\nu_\phi\in2\mathbb{Z}$ is doubled, then the quantization $\sigma_\text{Hall}\in \mathbb{Z}$ can occur at {\it two} distinct fractional values of the particle filling $\nu_\omega,\nu_\omega' = \nu_\omega+1$, where $\nu_\omega$ shares the same torsion $n$ as $\theta_\xi$. 

In other words, there are two distinct anomaly-free glide-symmetric FQH phases --- labeled by $\nu_\omega,\nu_\omega'$ --- with the same flux filling $\nu_\phi$, that quantizes the Hall conductance. This is once again consistent with the notion that weak glide-symmetric phases are precisely those that square to the trivial phase [\onlinecite{XA}].

\subsubsection{Screw-magnetic flux-filling invariants}
In the 3D screw case $P=P_{A;\nu}$ with $\nu=1\in\mathbb{Z}/\frac{n}{A}$ is the generator, the extension classes $f_k\in\operatorname{Ext}(\mathbb{Z}/k,\mathbb{Z}/\frac{n}{A}),f_m\in\operatorname{Ext}(\mathbb{Z}/m,\mathbb{Z}/\frac{n}{A})$ have order-$\frac{n}{A}$, which is the order of the quotient $C_n/C_A = \mathbb{Z}/\frac{n}{A}$. We may then immediately identify the fractionalization
\begin{equation}
\nu^1_\phi,\nu_\phi \in \frac{n}{A}\cdot \frac{1}{m}\mathbb{Z},\qquad \nu^1_\omega,\nu_\omega \in \frac{n}{A}\cdot \frac{1}{k}\mathbb{Z}.\nonumber
\end{equation}
Furthermore, there exist residual rotational symmetry $C_A\subset C_n$ characterized by the ordinary representations $V_k\in H^1(C_A,\mathbb{Z}/k)\cong \mathbb{Z}/\operatorname{gcd}(A,k),V_m\in H^1(C_A,\mathbb{Z}/m)\cong \mathbb{Z}/\operatorname{gcd}(A,m)$. These representations are assigned to each domain wall fixed by the point group $P'=1\times C_n$, hence within each fundamental domain there are $A$-number of Wyckoff positions occupied by these representations. 

To understand this structure, note that $H^{1;\nu}(C_n,\langle\tilde\Omega_\xi\rangle)$, as a model for $H^1(P_{A;\nu},\langle\tilde\Omega_\xi\rangle)$, is equipped with the restriction
\begin{equation}
\iota^*: H^{1;\nu}(C_n,\langle\tilde\Omega_\xi\rangle) \rightarrow H^{1}(\mathbb{Z}^d,\langle\tilde\Omega_\xi\rangle)\nonumber
\end{equation}
induced from the inclusion $\iota:\mathbb{Z}^d\cong\Lambda\hookrightarrow P_{A;\nu}$. As such each $W\in H^{1;\nu}(C_n,\langle\tilde\Omega_\xi\rangle)$ yields a representation $T = \iota^*W = (T_1,\dots,T_d) \in H^1(\mathbb{Z}^d,\langle\tilde\Omega_\xi\rangle)$ with the index $(\nu^1_\phi,\nu^1_\omega)$ fractionalized by $\frac{n}{A}$. On the other hand, these twisted representations must also pullback
\begin{equation}
S^*:H^{1;\nu}(C_n,\langle\tilde\Omega_\xi\rangle) \rightarrow H^{1}(C_A,\langle\tilde\Omega_\xi\rangle)\nonumber
\end{equation}
via the section $S:C_A\rightarrow P_{A;\nu}$ into a linear $C_A$-representation $V\in H^1(C_A,\langle\tilde\Omega_\xi\rangle)$. 

As $V$ is characterized by the indices $\mu_\phi \in \frac{1}{\operatorname{gcd}(m,A)}\mathbb{Z}, \mu\in\frac{1}{\operatorname{gcd}(k,A)}\mathbb{Z}$, we enforce
\begin{eqnarray}
\mu_\phi &=& \frac{m}{\operatorname{gcd}(m,A)} \nu_\phi \in \frac{1}{\operatorname{gcd}(m,A)}\mathbb{Z},\nonumber\\ \mu &=& \frac{k}{\operatorname{gcd}(k,A)}\nu_\omega \in \frac{1}{\operatorname{gcd}(k,A)}\mathbb{Z},\nonumber
\end{eqnarray}
which imposes $C_A$-equivariance on the restriction $T=\iota^*W$, such that $V=S^*W$ defines a linear $C_A$-representation. As such the torsions of $(\nu_\phi,\nu_\omega)$ are reduced to $m',k'$, respectively, where $$ m' = \frac{m}{\operatorname{gcd}(m,A)},\qquad k'=\frac{k}{\operatorname{gcd}(k,A)}.$$ We denote by $\nu_\phi',\nu_\omega'$ by the filling indices with these reduced torsions.

Now recall $c_\nu(\hat{g},\hat{g}')=\frac{A}{n}\cdot 1$ for $[g']=[g]+1\in C_n/C_A$, we take $\varphi=\frac{n}{A}\cdot \phi$ and normalized $\xi(\phi)^{-1}$ to $\xi(1)^{-1}$. The fraction $\theta_\xi \in\mathbb{Q}/\mathbb{Z}$ is then determined by $\xi$ with a denominator $\frac{A}{n}\operatorname{lcm}[k',m']$. We now arrive at the flux-filling invariant
\begin{equation}
-\sigma_\text{Hall} {\nu}_\phi' + {\nu}_\omega' - \theta_\xi \mod \frac{n}{A}\mathbb{Z}\label{eq:screwflux}
\end{equation}
corresponding to the class $(\tilde\Omega_\xi)_*h_\nu$, where we emphasize that $\nu_\phi',\nu_\omega'$ have $m',k'$-torsions modulo $\mathbb{Z}$, respectively. This Eq. (\ref{eq:screwflux}) gives rise to the new flux-filling constraint 
\begin{equation}
-\sigma_\text{Hall}{\nu}_\phi' +{\nu}_\omega' = \theta_\xi \mod \frac{n}{A}\mathbb{Z}, \nonumber
\end{equation}
which has no precedence in the physics literature. To apply such a result, we must find a 3D system that can exhibit the FQH effect.

It was found recently that the 3D Weyl semimetals is such a system [\onlinecite{CGB}]. Here, the global electromagnetic symmetry $U(1)_E$ as well as the translation symmetry $\mathbb{Z}=\mathbb{Z}[t_1]$ become anomalous, where the lattice vector $t_1$ is along the separation axis of the Weyl nodes. Take $K_p=SO(3),K_f=\mathbb{Z}/2$, we may describe such flux-dislocation defects with the composite ((0+1)+1)D phase $\tilde\Omega_\xi\in H^3(K,U(1))$, wherein a vortex loop $\ell$ of magnetic translations $\iota^*\gamma(\ell)$ traps a Majorana fermion $\omega_x=1$ at the intersection $x$ of $\ell$ with a plane along $t_1$.

Now fractionalization of the ground state comes from the two-loop braiding of such composite phases, whose statistics determines the torsion $s$ of $\theta_\xi$. If we suppose further that there is a $A_{\frac{n}{A}}$-screw along $t_1$, it is then possible to apply the anomaly-free condition of Eq. (\ref{eq:screwflux}) to describe the 3D fractional anomalous Hall effect on a screw-dislocated Weyl semimetal. In the particular case illustrated in Fig. \ref{fig:screwflux}, where $n=3$ and $A=1$, we would then expect a three-fold fractionalization of the flux-dislocation defect, {\it on top} of the semionic statistics [\onlinecite{CGB}].

\subsubsection{Glide-screw-magnetic flux-filling invariants}
In the glide-screw case, recall that the antiunitarity of the glide-screw operator $W_{\hat{g}}$ means that the indices $\bar\mu,\bar\mu_\phi\in\mathbb{Z}/2$ are only defined modulo two. This reduction means that there is no way to uniquely assemble twisted {\it antiunitary} representations $W\in H^{1;\nu}(C_2\times C_6,\langle\tilde\Omega_\xi\rangle_o)$ from a pair $(V;f)$ of a linear $C_3$-representation $V\in H^1(C_3,\langle\tilde\Omega_\xi\rangle)$ and the extension class $f\in \operatorname{Ext}(\langle t_1\rangle,\mathbb{Z}/2)$. This is an artifact of the coupling between a glide reflection with a 3-fold rotation in $P=\mathtt{gs}$.

Nevertheless, the diagonal subgroup $\mathbb{Z}/2\subset C_2 \times C_6/C_3$ has order two, and identifies the fractionalization
\begin{equation}
\nu_\phi^1,\nu_\phi\in 2\cdot\frac{1}{m}\mathbb{Z},\qquad \nu^1_\omega,\nu_\omega\in 2\cdot\frac{1}{k}\mathbb{Z}.\nonumber
\end{equation}
Furthermore, the non-symmorphic area cocycle $c_\nu([\hat{g}],[\hat{g}]+1)=\frac{1}{2}\cdot 1$ also computes the flux $\varphi =2\cdot \phi$, whence the fraction $\theta_\xi \in\mathbb{Q}/\mathbb{Z}$ given by $\xi$ is determined with a denominator $\frac{1}{2}\operatorname{lcm}[k,m]$. 

The ambiguity manifests instead in the pullback $S^*:H^{1;\nu}(C_n,\langle\tilde\Omega)_\xi\rangle)\rightarrow H^1(C_A,\langle\tilde\Omega_\xi\rangle)$ along the section $S$, which must factor through the inclusion
\begin{eqnarray}
\operatorname{Tor}(\mathbb{Z}/2,H)&\cong& \mathbb{Z}/\operatorname{gcd}(2,k)\oplus\mathbb{Z}/\operatorname{gcd}(2,m)\nonumber \\ 
&\hookrightarrow& H^1(C_3,\langle\tilde\Omega_\xi\rangle)\otimes\operatorname{Ext}(\langle t_1\rangle,\mathbb{Z}/2).\nonumber
\end{eqnarray}
The simplest way to satisfy this condition is for the linear $C_3$-representations $V\in H^1(C_A,\langle\tilde\Omega_\xi\rangle)$ to be trivial, with indices
\begin{equation}
\mu_\phi \in \frac{1}{\operatorname{gcd}(2,m,3)}\mathbb{Z}=\mathbb{Z},\qquad \mu \in \frac{1}{\operatorname{gcd}(2,k,3)}\mathbb{Z}=\mathbb{Z}.\nonumber
\end{equation}
This leads to the flux-filling invariant
\begin{equation}
-\sigma_\text{Hall}\nu_\phi+ \nu_\omega - \theta_\xi \mod 2\mathbb{Z}\label{eq:glidescrewflux}
\end{equation}
identical to Eq. (\ref{eq:glideflux}). 

Due to the antiunitarity of the glide-screw operator, we expect, for instance, the semionic braiding statistics of the flux-dislocation defect in 3D Weyl semimetals [\onlinecite{CGB}] to be further halved.

\section{Discussion}
In this paper, we have generalized {\bf Theorem \ref{thm:lsmanom}} [\onlinecite{ET}] by constructing LSM anomalies that take into account magnetic and non-symmorphic effects. We have also generalized the decorated domain wall construction [\onlinecite{RL}] to non-symmorphic lattices. 

Furthermore, by explicit computations, we have produced filling constraints on QI/QH systems based on anomaly-free conditions associated to the derived LSM anomaly. The results in Sec. \ref{sec:nsymfill} --- as well as Eq. (\ref{eq:glideflux}) --- agree with pre-existing results [\onlinecite{PWJZ}], [\onlinecite{LRO}], [\onlinecite{BBRFr}], [\onlinecite{BBRF}], [\onlinecite{WPVZ}], while those on the novel glide-screw lattice and Eq. (\ref{eq:screwflux}) are new derivations. We have also substantiated the heuristic argument [\onlinecite{WPVZ}] that the existence of a non-symmorphic equivariant $\text{Spin}_c$-structure [\onlinecite{GT}], [\onlinecite{GT1}] should non-trivially constrain the filling factor in {\it Remark \ref{rmk:latred}}.

According to the crystalline equivalence principle [\onlinecite{ET}], [\onlinecite{ET1}], [\onlinecite{ZYGQ}], anomaly-free invertible $G=P\times K$-SPT phases in $n$-dimensions is classified by the cohomology group $H^n(BP,E)$ [\onlinecite{CGLW}], [\onlinecite{GJF}], where $E$ is the spectrum of invertible $K$-SPT phases. The in-cohomology $K$-protected phases are classified by [\onlinecite{RL}] $$\pi_0\mathcal{E}^\text{in} = \bigoplus_{l\leq d}\pi_0\mathcal{E}^{\text{in}}_{d-l}\cong \bigoplus_{l\leq d}H^{l}(P,H^{d-l+2}(K,U(1))).$$ The factors at $l < d$ detect higher order $K$-symmetric topological phases localized on $P$-equivariant subspaces of codimension $l$ [\onlinecite{LiWe}]. This perspective forms the basis of the decorated domain walls construction [\onlinecite{RL}], as well as our generalization of it. 

Its relation to the filling invariants can be understood as follows. The crystalline $K$-SPT phase $\mathcal{D}^{d+2}_P\cong\operatorname{Bulk}\mathcal{C}^{d+1}_P$ constructed in Sec. \ref{sec:lsmqshe} is first order, and as such is classified by the $l=d$-th term $\pi_0\mathcal{E}_0$ indexed by merely the particle filling factor $\nu_\omega$. By hypothesis, this summand $\mathcal{E}_0$ is protected by an internal electric $K_p$-symmetry. 

In contrast, the second order crystalline $K$-SPT phase $\mathcal{G}^{d+2}_P\cong\operatorname{Bulk}\mathcal{F}^{d+1}_P$ constructed in Secs. \ref{sec:lsmfqhe}, \ref{sec:nsymmag} is classified by the two $l=d,d-1$ terms $\pi_0\mathcal{E}_0\oplus\pi_0\mathcal{E}_1$, in which the latter summand is protected by an internal magnetic $K_f$-symmetry. Correspondingly, we have shown that such second order magnetic phases are indexed by a tuple $(\nu_\phi,\nu_\omega)$ characterized by Eq. (\ref{eq:fluxfill}), with the index $\theta_\xi$ describing the "interaction" between the phases at degrees $l=d,d-1$.

\subsection{Future Work}
We have taken numerous simplifying assumptions in our above results. First, we have assumed that the lattice $\Lambda$ is crystalline, with space group $P$ classified by a class $\nu\in H^2(P',\Lambda)$. In this case, the crystallographic restriction theorem (see Appendix \ref{sec:spacegrp}) holds and the non-symmorphic effects are straightforward. In general, however, there exist {\bf quasicrystalline} lattices whose structures are ordered but not periodic. 

In the absence of point group symmetry, a quasicrystalline equivalence was shown [\onlinecite{EHPG}] such that the in-cohomology classification is accomplished with cohomology groups on the "elasticity fundamental domain" $\Gamma'\cong \mathbb{T}^D = B\mathbb{Z}^D$. It would then be reasonable to expect non-symmorphic twists to manifest as $\nu\in H^2(P',\mathbb{Z}^D)$ in the quasicrystal context, but its exact structure would require further, more detailed study.

% In Ref. [\onlinecite{EHPG}, it was shown that invertible quasicrystalline phases are classified instead by the {\it non-connective} spectrum
% \begin{equation}
% \mathcal{E}' \cong \bigoplus_{l\leq D}\mathcal{E}_{d-l}^{\times \binom{D}{l}},\nonumber
% \end{equation}
% where $D$ labels the number of elastic modes for which $D>d$ in a quasicrystal. In the absence of point group symmetry $P'$, it was shown that we may express [\onlinecite{EHPG}
% \begin{equation}
% \mathcal{E}' \cong \bigoplus_{l\leq D} H^{l}(\Gamma',E_{d-l}) \nonumber
% \end{equation}
% in terms of the spectrum $E$ of $K$-SPT phases, where $\Gamma'\cong \mathbb{T}^D = B\mathbb{Z}^D$ is the "elasticity fundamental domain". 

Second, we have merely focused on "in-cohomology" SPT phases classified by group cohomology. There exist out-cohomology SPT phases, such as the 2D $E_8$-phase [\onlinecite{GJF}], [\onlinecite{RL}], that lie outside of the group cohomology formalism. Out-cohomology equivariant classes $h\in H^3_{P'}(X,\mathbb{Z})$ also appear in "partial Fourier-Mukai transforms" that fit into a $T$-duality web [\onlinecite{SSG}], [\onlinecite{GT1}],
% \begin{equation}
% \begin{tikzcd}
%                                               &                & K^{\ast-1;\sigma_P+h_S}_{P'}(X\times\widehat{S}) \arrow[rrd,"\sim"] & {} &                                             \\
% K^{\ast;\sigma_P}_{P'}(X\times S) \arrow[rrd,"\sim"] \arrow[rru, "\sim"] \arrow[rrrr,"\mathscr{T}"] & {}  &                                                    &    & K^{\ast;\tau_P}_{P'}(\widehat{X\times S}) \\
%                                               &                & K^{\ast-1;\sigma_P+h_X}_{P'}(\widehat{X}\times{S}) \arrow[rru,"\sim"]&    &                                            
% \end{tikzcd}, \nonumber
% \end{equation}
and are in fact related to the $h$-fluxes encountered in string theory [\onlinecite{AP}]. As they do not appear in the full $T$-duality map, these classes specifically do not seem to affect the classification of SPT phases.

Thirdly, we have assumed the conditions necessary to split the second order SPT phase $\Omega\in H^3(K,U(1))$ into "electric" ($\in H^2(K_p,U(1))$) and "magnetic" ($\in H^1(K_f,U(1))$) parts. In the generic case, we expect non-trivial mixture between them, whence the filling invariant is not merely a sum of the particle filling $\nu_\omega$ and the flux filling $\nu_\phi$. We shall tackle this issue by constructing an appropriate effective 2-gauge TQFT $Z_\text{QMag}$ [\onlinecite{KT}] in a future work.

\begin{remark}
We expect to extract a generalized notion of a "flux-filling invariant" from the partition function $Z_\text{QMag}(X)$ on a generic space $X$. For the sphere $X=S^2$, in particular, this would put the {\it shift index} $\mathcal{S}=\nu_\omega-\nu_\phi$ on a topological footing. Furthermore, $\mathcal{S}$ appears in a sum rule for the density $\expval{\rho}$ of the {\bf magnetoroton} excitation [\onlinecite{GNS}], which is a dynamical quantity with no clear topological interpretation yet.
\end{remark}

It would also be important to develop a classification of fermionic crystalline invertible SPT phases in the non-symmorphic context. This would allow us to derive, for instance, a set of filling constraints for superconductors. As a "first approximation", one may consider coupling first order phases $\mathcal{E}_0$ to the fermion parity symmetry $\mathbb{Z}/2^F$  [\onlinecite{KitMajor}]. More generally, however, it was argued [\onlinecite{Deb}] that the fermionic SOC must be part of the data of an equivariant local system $\mathcal{L}$ on the spectrum $\mathcal{E}$. One would then have to also study the effect of non-symmorphic twists $\nu\in H^2(P',\Lambda)$ on $\mathcal{L}$.

\begin{acknowledgments}
The author would like to thank Davide Gaiotto and Theo Johnson-Freyd for valuable discussions, and the University of Waterloo for providing accomodation and resources.
\end{acknowledgments}

\newpage

\appendix

\section{Classification of Space Groups}\label{sec:spacegrp}
Let $C_n,D_n$ denote the cyclic/dihedral groups of order $n/2n$, respectively. We are interested in, rather than the trivial $\mathbb{Z}[P']$-module $\mathbb{Z}^d$, the {\bf Bravais lattice} $\Lambda$, which admits an orthogonal automorphism by $P'$ through the embedding $\rho: P'\hookrightarrow O(d)$. As Abelian groups $\Lambda\cong \mathbb{Z}^d$, but $\Lambda$ also has the structure of an {\it integral representation space} of $P'$ [\onlinecite{AP}]. This fact plays a major role in the classification of space groups.

Recall [\onlinecite{Hill}], [\onlinecite{Dav}] that the classification of {\it point} groups $P'$ in $d$ spatial dimensions identifies it as a finite group with presentation
\begin{equation}
P' = \langle a_1,\dots,a_g\mid r_1,\dots,r_t\rangle,\nonumber
\end{equation}
with given relations $r_1,\dots,r_t$. Let $\rho:P'\rightarrow \operatorname{GL}_d(\mathbb{Z})$ be an integral representation on the lattice $\Lambda\cong\mathbb{Z}^d$ as a $\mathbb{Z}[P']$-module. Consider the surjection $F_g=\operatorname{Free}(a_1,\dots,a_g)\rightarrow P'$ from free group by projecting out the relations $R_t = \langle r_1,\dots,r_t\rangle$, we write elements in its kernel
\begin{equation}
r_i -1 = \sum_{j=1}^g \Sigma_{ij}(a_j-1),\qquad \Sigma_{ij}\in \mathbb{Z}[F_g]\nonumber
\end{equation}
by a change of basis in the $\mathbb{Z}[F_g]$-algebra. Following a procedure due to Zassenhaus, we have  
\begin{theorem}
\label{thm:zass}
Define the $dg\times dt$ integral matrix $\rho\Sigma \in M_{dg\times dt}(\mathbb{Z})$, and let $b_1,\dots,b_u\in\mathbb{Z}$ denote its non-trivial invariant factors such that $b_k|b_{k+1}$ for all $1\leq k< u$, then
\begin{equation}
H^2(P',\Lambda) \cong \bigoplus_{k=1}^u \mathbb{Z}/b_k.\nonumber
\end{equation}
Normalizers $\alpha\in N=N_{\operatorname{GL}_d(\mathbb{Z})}(P')$ act on $H^2(P',\Lambda)$ via the commutative diagram
\begin{equation}
\begin{tikzcd}
\mathcal{E}: \arrow[d, maps to] & 1 \arrow[r] & \Lambda \arrow[r, "\iota"] \arrow[d, "\alpha"] & P \arrow[r, "p"] \arrow[d, "\cong"]         & P' \arrow[r] \arrow[d, "\operatorname{Conj}_\alpha"] & 1 \\
\alpha\mathcal{E}:              & 1 \arrow[r] & \Lambda \arrow[r]          & P \arrow[r] & P' \arrow[r]                                         & 1
\end{tikzcd},\nonumber
\end{equation}
whence affine equivalent classes of space groups $P$ are in one-to-one correspondence with $H^2(P,\Lambda)/\operatorname{Orb}_N$.
\end{theorem}
\noindent The point of computing extension classes $H^2(P',\Lambda)$ is that it gives the differentials $d_2$ in the second page $E_2$ of the Lyndon-Serre-Hochschild (LHS) spectral sequence, which we shall study later.

\subsection{Low Dimensional Space Groups}
At low degrees, all possible Bravais space groups $P$ have been classified. In case $P'=C_n$ is cyclic of order $n$, we have
\begin{equation}
H^2(C_n,\Lambda) \cong \Lambda^{C_n}/\sum_{p\in C_n}p\cdot \Lambda\nonumber
\end{equation}
where $\Lambda^{P'}\subset\Lambda$ denotes the sublattice of fixed points.  By the {\bf crystallographic restriction theorem}, only those with $n=2,3,4,6$, classifying half-plane, trigonal, tetragonal and hexagonal tilings, are Bravais in 2D. Furthemore, no non-symmorphic effects can occur in case $n=3,4,6$, as no fixed points are present.

However, when $n=2$ we have $H^2(C_2,\Lambda)\cong \mathbb{Z}/2$ classifying a glide/reflectotranslation element. We may then form the lattices
\begin{equation}
\Lambda^0 = \mathbb{Z}\oplus\mathbb{Z},\qquad \underbrace{\Lambda^1}_{\mathtt{pg}} = \mathbb{Z}\oplus \mathbb{Z}',\qquad \underbrace{\Lambda^2}_{\mathtt{pgg}} = \mathbb{Z}'\oplus\mathbb{Z}',\nonumber
\end{equation}
where the generator of $C_2= \mathbb{Z}/2$ acts as $-1$ in $\mathbb{Z}'$. The non-trivial lattices $\Lambda^1,\Lambda^2$ each contain one and two glide axes, respectively; see Fig \ref{fig:glides}.

\begin{figure}[h]
    \centering
    \includegraphics[width=1\columnwidth]{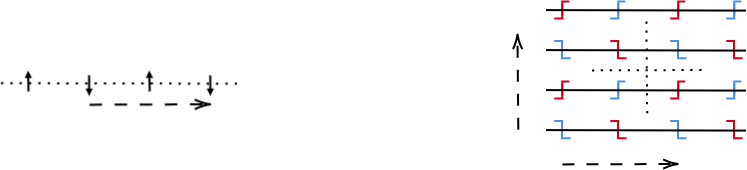}
    \caption{Case $P'=C_2$. {\bf Left}: the lattice $\Lambda^1$ with one glide ($\mathtt{pg}$). {\bf Right}: the lattice $\Lambda^2$ with two glides ($\mathtt{pgg}$). Dotted lines denote reflection axes, and dashed arrows denote residual translations.}
    \label{fig:glides}
\end{figure}

In accordance with the classification of 2D non-symmorphic space groups above, for $P'=C_2$ we have the following table
\begin{center}
\begin{tabular}{|c c c|}
\hline 
\# of Glides & Lattice & Unit Cell $BP$ \\
\hline
0 & $\Lambda^0$ & 2-Torus $\mathbb{T}^2$ \\ 
1 & $\Lambda^1$ & Klein bottle $\mathbb{K}$ \\ 
2 & $\Lambda^2$ & Real-projective plane $\mathbb{R}P^2$ \\
\hline
\end{tabular}
\end{center}
The class $\nu\in H^2(P',\Lambda)$ manifests as the non-orientability of $BP$. Take $P=\mathtt{pg}$, for instance, as a orbifold quotient of the twisted toroidal fibration $B\mathbb{Z}'\hookrightarrow B\Lambda^1\xrightarrow{p} B\mathbb{Z}$ [\onlinecite{SSG}], the glide element identifies a point $\lambda \in p^{-1}(m)$ in the fibre over $m\in B\mathbb{Z}$ with its reflection $-\lambda \in p^{-1}(m+\pi)$ over the translate $m+\pi\in B\mathbb{Z}$.

\subsection{3D Layer Space Groups}
The space group $P\rightarrow P'$ whose point group is axial $P'= P'_1\times P'_2$ is called a {\bf layer group}. Here Bravais lattices $\Lambda$, as a $\mathbb{Z}[P']$-module, decomposes into a $\mathbb{Z}[P'_2]\oplus\mathbb{Z}[P'_1]$-module such that the representation $\rho=\rho_1\oplus \rho_2:P'\rightarrow \operatorname{GL}_1(\mathbb{Z})\oplus\operatorname{GL}_2(\mathbb{Z})\hookrightarrow\operatorname{GL}_3(\mathbb{Z})$ splits. Given the generator $\alpha\in P'_1=D_1=\langle \alpha\mid \alpha^2\rangle$ in the only non-trivial 1D point group, it can modify the relations $R\rightarrow R(\alpha)$ in $P'_2=\langle a_1,\dots,a_g\mid r_1,\dots r_t\rangle$ by replacing each instance of, say, $a_1\in P'_2$ by $\alpha a_1\in P'$. This leads to either a direct or a semidirect $\mathbb{Z}/2$-product, and the following list
\begin{widetext}
\begin{eqnarray}
P_1' = 1 &\times& \begin{cases}P'_2 = C_n \\ P_2' = D_n\end{cases} \implies \begin{cases}P'=C_n &; \text{rotational $C_n$} \\ P'=D_n &; \text{pyramidal $C_{nv}$} \end{cases},\nonumber \\
P_1' = \mathbb{Z}/2&\times& \begin{cases}P'_2 = C_n \\ P'_2=D_n\end{cases} \implies \begin{cases}P'=\mathbb{Z}/2\times C_{n} &; \text{chirotational $C_{nh}$} \\ P' = \mathbb{Z}/2\times D_n &; \text{prismatic $D_{nh}$} \end{cases},\nonumber \\
P_1' = \mathbb{Z}/2 &\ltimes& \begin{cases}P'_2 = C_n \\ P'_2 = D_n\end{cases} \implies \begin{cases} P' = C_{2n} \qquad ; \text{rotoreflection $S_{2n}$} \\ \begin{cases} P' = D_{2n} &; \text{antiprismatic $D_{nd}$} \\ P' = D_n &; \text{dihedral $D_n$}\end{cases}\end{cases}\nonumber
\end{eqnarray}
\end{widetext}
of the seven infinite families of axial point groups; in the dihedral case $P'_2=D_n$, we have the choice of coupling $\alpha$ to a rotation or a reflection generator in the dihedral group $D_n$.

First consider the example of the rotational $P'=1\times C_n$ case, whose 2D cyclic subgroup $P'_2=C_n$ of order $n=2,3,4,6$ acts only on the 2D sublattice plane $\Lambda_2\subset\Lambda$. Here the generator $\alpha=1$ is trivial, and each integral representation $\rho=\rho_A$ is classified by a divisor $A|n$ for which $\rho_2(a_1^A)= \operatorname{id}$ and $\operatorname{gcd}(\frac{n}{A},A)=1$. With the invariant axis $\Lambda_1 \perp \Lambda_2$, the action $\rho_A(1,a_1)\Lambda_1=a_1\cdot_A\Lambda_1$ then gives
\begin{equation}
H^2(C_n,\Lambda) \cong \Lambda^{C_n}/\sum_{k=1}^na_1^n \cdot_A \Lambda  = \Lambda_1/\frac{n}{A}\Lambda_1 \cong \mathbb{Z}/\frac{n}{A},\nonumber
\end{equation}
characterizing a screw/rototranslation lattice $\Lambda=\mathbb{Z}^2\oplus \Lambda_A$. The value $\nu\in\mathbb{Z}/\frac{n}{A}$ counts the number of translations that an $A$-fold rotation fits in, and as such the order $\operatorname{ord}\nu$ counts the length of the screw. The space group $P=P_{A;\nu}$ then acquires an $A_{\operatorname{ord}\nu}$-screw along $\Lambda_A$. A similar argument holds in the pyramidal case $P'=1\times D_n=C_{nv}$, whence $\nu\in H^2(C_{nv},\Lambda)\cong \mathbb{Z}/\frac{n}{A}$ characterizes $A_{\operatorname{ord}\nu}$-screws along $\Lambda_1$. The presence of a reflection element in $D_n$, however, forces this screw to be axialsymmetric; see Fig. \ref{fig:screws}.

\begin{center}
\begin{figure}[h]
    \centering
    \includegraphics[width=1\columnwidth]{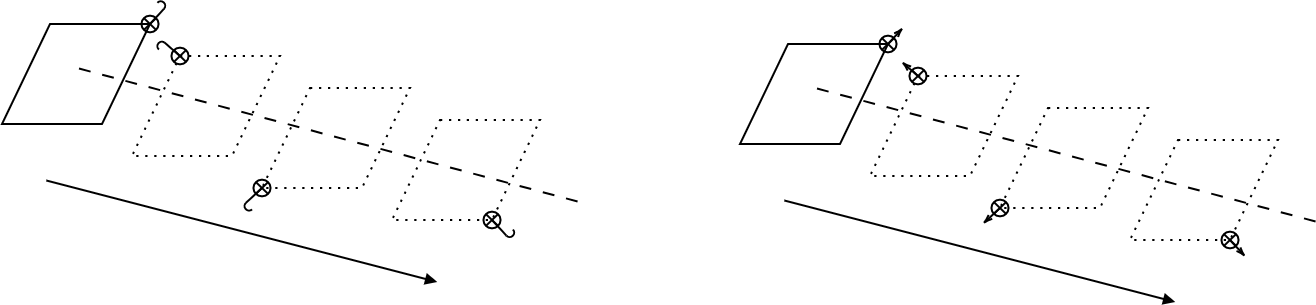}
    \caption{Screw patterns with $n=4$, $A=1$ and $\nu=1$. {\bf Left}: A non-axialsymmetric screw with $P'=1\times C_4$. {\bf Right}: An axialsymmetric screw with $P'=1\times D_4$. Notice the symmetry of the arrowheads.}
    \label{fig:screws}
\end{figure}
\end{center}

Now suppose $\alpha\neq 1$ is non-trivial. Putting $\alpha$ in the 0th position of the presentation for $P'$, we may generically write
\begin{equation}
P' = \begin{cases}\langle \alpha,a_1,\dots,a_g\mid R(\alpha)\rangle &; \cong  \mathbb{Z}/2\ltimes P'_2 \\ \langle \alpha,a_1,\dots,a_g\mid R(\alpha), c(\alpha)\rangle &; \cong \mathbb{Z}/2\times P'_2\end{cases},\nonumber
\end{equation}
where $c(\alpha)=\alpha a_1\alpha a_1^{-1}$ is the commutator relation. We may then compute the modified entires $\tilde{\Sigma}_{ij}$ such that $\rho\tilde{\Sigma}\in M_{3(t+1)\times 3(g+1)}(\mathbb{Z})$. The zero-th relation reads $\alpha^2=1$, which gives $\tilde{\Sigma}_{00} = \alpha+1$. Depending on whether $\rho_1(\alpha)=\pm1$, we obtain the zero-th invariant factor $\tilde{b}_0=0,2$. In case $\tilde{b}_0$ is non-trivial, all subsequent invariant factors $\tilde{b}_k$ must be even. This tells us that odd-order non-symmorphic effects cannot coexist with glides.

\begin{proposition}
The {\bf unique} glide-screw lattice $\Lambda$ is classified by $\nu =(1,1)\in H^2(P',\Lambda)$, where $P'$ is the chirorotational point group $P'=C_{6h}$ of order $6$ and the integral representation $\rho_A$ is characterized by $A=3$.
\end{proposition}
\begin{proof}
If $P'_1=1$ is trivial, then there are no reflections that may generate glides, hence let $P'_1=\mathbb{Z}/2$ be non-trivial. In the rotoreflection case $P' = \mathbb{Z}/2\ltimes C_n = C_{2n}$, the reflection generator $\alpha\in P'_1$ couples to the rotation generator $a_1\in C_n$ such that $\alpha a_1\in P'$ generates the rotoreflection. This element only fixes the origin $0\in\Lambda$, hence $\Lambda^{P'}=0$ and $H^2(C_{2n},\Lambda) = 0$.

Consider now the chirorotational case $P' =\mathbb{Z}/2\times C_n=C_{nh}$. The additional commutator relation $c(\alpha)=\alpha a_1\alpha a_1^{-1}$ decouples the lattice $\Lambda=\Lambda_1\oplus\Lambda_{23}$ as a $C_{nh}$-module, and splits the representation $\rho=\rho_1\oplus\rho_2$. Now non-trivial fixed points exist only when $n$ is even, for which $\rho_2: \sigma=a_2^{\frac{n}{2}} \mapsto \operatorname{id}$. Such representations $\rho_2$ are characterized by an even divisor $A=2A'$, and $\rho_2$ factors through a representation $\rho_{A'}$ of $C_n/C_2 = C_{n/2}$.

Given this axis, say $\Lambda_2\subset\Lambda$, of fixed points as a trivial $C_{nh}$-module, we may then apply K{\" u}nneth formula 
\begin{eqnarray}
H^2(C_{nh},\Lambda) &\cong& \bigoplus_{p+q=2}H^p(C_2,\Lambda_2) \otimes H^q(C_n,\Lambda_2)\nonumber\\ 
&=& H^2(C_2,\Lambda_2) \oplus H^2(C_n,\Lambda_2), \nonumber
\end{eqnarray}
as cyclic groups $C_n$ lack a (free) degree-1 cohomology generator [\onlinecite{Hill}], [\onlinecite{Dav}]. The previous results then give
\begin{equation}
H^2(C_{nh},\Lambda) \cong  \mathbb{Z}/\tilde{b}_0 \oplus\mathbb{Z}/\tilde{b}_1= \underbrace{\mathbb{Z}/2}_{\text{glide}} \oplus \underbrace{\mathbb{Z}/\frac{n}{A'}}_{\text{screw}}. \nonumber
\end{equation}
The glide-screw lattice $\Lambda$ is classified by a non-trivial class $\nu=\nu_1\oplus\nu_2\in H^2(C_{nh},\Lambda)$ for which $\nu_1=1,\nu_2\neq 0$. Such a non-trivial class $\nu=(1,1)\in \mathbb{Z}/2\oplus\mathbb{Z}/2$ exists only when $n=6$ and $A'=3$.
\end{proof}
\noindent Notice if $n=2$ and $A'=1$, then we recover layers of the usual glide lattice $P=\mathbb{Z}\times \mathtt{pg}$. An example of a glide-screw pattern is shown in Fig. \ref{fig:glidescrew}. The classification of the rest of the 3D layer space groups can analogously be carried out algorithmically with, for instance, GAP. 

\begin{center}
\begin{figure}[h]
    \centering
    \includegraphics[width=1\columnwidth]{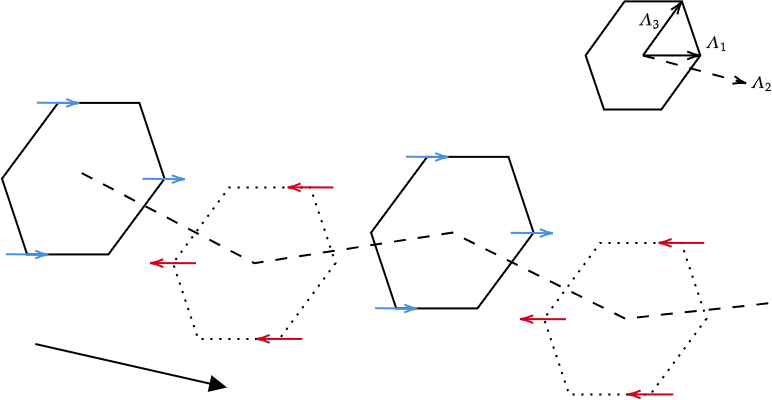}
    \caption{A glide-screw combination with $n=6$ and $A'=3$. Here $P'=C_2\times C_6$.}
    \label{fig:glidescrew}
\end{figure}
\end{center}

\subsection{3D Non-Layer Space Groups}
In this section, we study {\bf non-layer} space groups $P\rightarrow P'$, whose point groups $P'$ are non-axial. These consist of the tetrahedral, octahedral and icosahedral familities, together with their chiral counterparts (direct product with $\mathbb{Z}/2$). Integral representations $\rho:P'\rightarrow\operatorname{GL}_3(\mathbb{Z})$ do not split $\operatorname{GL}_1(\mathbb{Z})\oplus\operatorname{GL}_2(\mathbb{Z})$, as such we must compute the invariant factors $b_k$ from scratch.

\begin{proposition}
Within the tetrahedral family $P'=A_4$ (chiral tetrahedral), $P'=S_4$ (tetrahedral), $P'=A_4\times\mathbb{Z}/2$ (pyritohedral), only $S_4$ has a non-symmorphic glide element.
\end{proposition}
\begin{proof}
The strategy is as follows: we first construct the tetrahedral lattice $\Lambda_t$ and its associated representation $\rho_t$ by which $\rho_t(P')$ act as isometries. We then derive the relational matrices in each of the above groups and compute their invariant factors with help from a computer algebra program.

First, the chiral tetrahedral point group 
\begin{equation}
P'=A_4 = \langle a_1,a_2\mid a_1^3, a_2^3,(a_1a_2)^2\rangle\nonumber
\end{equation}
has the relational matrices
\begin{eqnarray}
\Sigma_{11} = a_1^2 + a_1 + 1,&\quad& \Sigma_{22}=a_2^2+a_2+1,\nonumber\\
\Sigma_{31} = a_1a_2+1,&\quad& \Sigma_{32} = (a_1a_2+1)a_1, \nonumber
\end{eqnarray}
with the other entries zero. Now suppose the vertices $V=\{t_1,t_2,t_3,t_4\}\subset \Lambda_t$ satisfying $t_4=-(t_1+t_2+t_3)$ span a regular tetrahedron, then the generators $a_1,a_2\in P'$ act as by 3-cycle permutations on $V$ through the tetrahedral representation $\rho_t$. We compute the matrix $\rho_t\Sigma \in M_{6\times 9}(\mathbb{Z})$ and its invariant factors $b_1 =b_2=\dots=b_6 =1$, as such 
\begin{equation}
H^2(A_4,\Lambda_t) = 0\nonumber
\end{equation}
and no non-symmorphic effects can manifest on $\Lambda_t$.

Next, the non-chiral tetrahedral group
\begin{equation}
S_4=\langle a_1,a_2\mid a_1^2,a_2^4,(a_1a_2)^3\rangle \label{eq:sym4}
\end{equation}
has relational matrices
\begin{eqnarray}
\Sigma_{11}=a_1+1,&\quad& \Sigma_{22} = a_2^3+a_2^2+a_2+1,\nonumber \\
\Sigma_{31}=(a_1a_2)^2+a_1a_2+1,&\quad& \Sigma_{32}=((a_1a_2)^2+a_1a_2+1)a_1. \nonumber
\end{eqnarray}
In the tetrahedral representation $\rho_t$ for which $a_1$ is a mirror reflection while $a_2$ is a 3-cycle, we compute the invariant factors $b_1=b_2=\dots=b_5=1,b_6=2$. Hence there is a glide element
\begin{equation}
H^2(S_4,\Lambda_t) \cong  \mathbb{Z}/2 \nonumber
\end{equation}
along the intersection $M_1\cap M_2\cap M_3$ of the mirror reflection planes $M_{1,2,3}$. 

Lastly, under the tetrahedral representation $\rho_t$, each 3-cycle permutation within the pyritohedral group $A_4\times\mathbb{Z}/2$ has an inversion-partner, each of which only fixes the origin of $\Lambda_t$. As such an application of K{\" u}nneth formula yields
\begin{equation}
H^2(A_4\times\mathbb{Z}/2,\Lambda_t) =0\nonumber
\end{equation}
by the computation of $H^2(A_4,\Lambda_t)=0$ above. Indeed, mirror planes of the pyritohedron intersect trivially.
\end{proof}
\noindent We can in fact say more: in symmorphic non-layer families $P\cong \Lambda\rtimes P'$, we actually have K{\" u}nneth formula
\begin{equation}
H^n(P,\mathbb{Z}) \cong \bigoplus_{p+q=n}H^p(P',H^q(\Lambda,\mathbb{Z})), \nonumber
\end{equation}
as the lattice $\Lambda$ must necessarily be a permutation module [\onlinecite{AP}]. This allows us to bypass any spectral sequence computations for the lattice anomaly $H^d(P,\mathbb{Z})$.

In general, non-symmorphic effects are "rare": the $\mathbb{Z}[P']$-module $\Lambda$ must be chirality-free, and mirror planes must intersect non-trivially. It is reasonable to expect the octahedral family to all be symmorphic on the octahedral lattice $\Lambda_o$, as the octahedral group $P'=S_4$ have no mirror planes while those in the chiral octahedral group $P'=S_4\times\mathbb{Z}/2$ intersect trivially. What is perhaps most interesting is the icosahedral family $P'=A_5\times\mathbb{Z}/2,A_5$; computing $H^2(P',\Lambda_i)$ would require one to write down an icosahedral lattice $\Lambda_i$ (up to $\mathbb{Z}[P']$-isomorphism), which has 12 vertices.

\section{Cohomology of Crystallographic Space Groups}\label{sec:lhs}
In this section, we lay out some basic properties of the group cohomology of crystallographic point groups in 2D and 3D. In particular, we examine some basic properties of their generators in $d=1,2,3$. 

\subsection{Symmorphic Lattices: Euler Class of the Point Group}
We have seen in the main text that if $P\cong \mathbb{Z}^d\rtimes P'$ is symmorphic, then we may classify (bosonic) crystalline SPT phases by computing classes that live on the {\it commutative} fundamental domain $\Gamma=B\mathbb{Z}^d$, then imposing $P'$-equivariance. Given the $O(d)$-Euler class on $\Gamma$, this can be achieved by pulling-back $e(P')=\rho^*e_d\in H^d_{P'}(\Gamma,\mathbb{Z}_o)$ via the linear embedding $\rho:P'\hookrightarrow O(d)$. Here we give a few examples of $e(P')$ in low dimensions.

\begin{enumerate}
\item $d=1$: Here, $\operatorname{Isom}\mathbb{Z} = \mathbb{Z} \times \mathbb{Z}/2$ is a product of translations with the inversion group $\mathbb{Z}/2$ (cf. Bieberbach's theorem in $d=1$). Given an irrep $V$ of $\mathbb{Z}/2$, the quantity $r=\operatorname{det}V\in \mathbb{Z}/2$ determines the orientation of $V$, and hence grades the coefficients of the cohomology $H^\ast(\mathbb{Z}/2,\mathbb{Z}_o)$ [\onlinecite{GT}], [\onlinecite{GT1}], [\onlinecite{SSG}]. 

\item $d=2$: Here $P'$ is either cyclic or dihedral of fixed degree $p$. Given an irrep $V$ of $P'$, the Euler class $e(P') \in H^2(P',\mathbb{Z}_o)$ determines its isomorphism class. Now as $P'$ has order $p$, the Euler class is $p$-torsion $p\cdot e(P') = 0$. If we denote by $\alpha\in C^1(P',\mathbb{Z}_o)$ its 1-cochain trivialization, we then have
\begin{equation}
e(P') = \frac{1}{p}d_o\alpha = \frac{1}{p}\left( d\alpha - 2r\cup \alpha\right) \in H^2(P',\mathbb{Z}_o). \nonumber
\end{equation}
If $P' = C_p$ is cyclic, then $r= 0$ and $d_o = d$ in this case.

\item $d=3$: the 2D point group $P'_2$ embeds into the full 3D point group $P'$. When $P$ is axial --- namely when $P$ contains a 1D {\it Frieze group} $P'_1$ along an axis --- the irrep bundle $V$ decomposes accordingly, and so does the Euler class [\onlinecite{ET}]
\begin{eqnarray}
e(P') &=& \sigma\cup e(P'_2) \in H^3(P',\mathbb{Z}_{o+\sigma_2}),\nonumber\\
\sigma &=& \begin{cases}\alpha &; \text{rotation} \\ r &; \text{reflection}\end{cases}\in H^1(P'_1,\mathbb{Z}).\label{eq:3dpg}
\end{eqnarray}
Here, the twist $o$ comes from a reflection element in $P'_2$, while $\sigma_2$ denotes the mod-2 reduction of $\sigma$.
\end{enumerate}

More generally, mod-2 reduction of the Euler class can be obtained via the {\it Steenrod square} $\operatorname{Sq}^j:H^\ast(\bullet,\mathbb{Z}/2) \rightarrow H^{\ast+j}(\bullet,\mathbb{Z}/2)$ through the {\bf Wu forumula} [\onlinecite{ET}], [\onlinecite{SP}], [\onlinecite{Hat}]
\begin{equation}
\operatorname{Sq}^j w_i = \sum_{k=0}^i\binom{(j-i)+(k-1)}{k}w_{i-k}\cup w_{j+k}, \nonumber
\end{equation}
on Stiefel-Witney classes $w_i$. In general, the Euler class generates the image $\operatorname{im}\mathcal{B}$ of the Bockstein map $\mathcal{B}:H^\ast(\bullet,\mathbb{Z}/2)\rightarrow H^{\ast+1}(\bullet,\mathbb{Z})$ [\onlinecite{Fesh}], hence one may write $e(P')=\mathcal{B}w_2(P')$ directly in $d=3$ [\onlinecite{ET}].

In the non-symmorphic case, however, all of this technology breaks down as the affine action of $P$ on $X$ does not split into a free action of $\mathbb{Z}^d$ and a linear action of $P'$. As such we cannot leverage the Euler class $e(P')$, but instead must compute non-symmorphic classes in $H^d(P,\mathbb{Z}_o)$ directly.

\subsection{Structure of the Spectral Sequence}
All cohomology groups $H(\bullet)$ in the following have integral coefficients $\mathbb{Z}$, unless otherwise specified. Given the central extension sequence that $P$ fits in, we may compute $H^d(P)$ with the {\bf Lyndon-Hochschild-Serre (LHS) spectral sequence}
\begin{equation}
 H^p(P',H^q(\Lambda))\cong E^{p,q}_2\Rightarrow H^{d}(P).\nonumber
\end{equation}
The differential $d_r: E^{p,q}_r\rightarrow E^{p+r,q-r+1}_r$ at each page $E_r$ computes entries of the next page $E_{r+1}^{p,q} = \operatorname{ker}d_r/\operatorname{im}d_r$. Those entries $E^{p,q}_r$ that are unaffected by $d_r$ (i.e. at $p<r,q<r-1$) survive to the stable limit $E_\infty^{p,q}$; in particular, we get
\begin{equation}
E^{0,0}_\infty=\mathbb{Z},\qquad E^{1,0}_\infty = H^1(P') \nonumber
\end{equation}
for free. We say that the spectral sequence $E_r$ degenerates at page $r$ if $d_{r'}=0$ for all $r'\geq r$, whence $E_r = E_\infty$ nets us the graded objects $E_r^{p,d-p}=\operatorname{gr}_pH^d(P)$ of a filtration of $H^d(P)$ [\onlinecite{Hat}].

Here we use the fact that $\Lambda$ is Bravais to write $H^\ast(\Lambda)\cong \mathbb{Z}[t_1,\dots,t_d]$ as a ring generated by the duals of the primitive lattice vectors $\{t_\alpha\}_{\alpha\leq d}$. The $q$-th row $E^{\ast,q}_2$ of the second page then consist of $\binom{d}{q}$-copies of the $0$-th row $E^{0,q}_2\cong H^p(P')$, with a generic class written as sums of
\begin{equation}
e_p\otimes_\rho t_{i_1}\cup\dots\cup t_{i_q} \in H^p(P',H^q(\Lambda)),\nonumber
\end{equation}
where
\begin{equation}
    e_p\in H^p(P'),\qquad 1\leq i_1< \dots< i_q\leq d.\nonumber
\end{equation}
As such the LHS spectral sequence must degenerate at $E_{d+2}$ at the latest, and we may assemble the graded objects $\operatorname{gr}_pH^d(P) = E^{p,d-p}_{d+2}$ from entries of that page.

\subsubsection{Cohomology of the 2D glide lattice}
Consider the glide lattice $\Lambda = \mathbb{Z}\oplus\mathbb{Z}'$ characterized by the non-trivial class $\nu\in H^2(C_2,\Lambda)= \mathbb{Z}/2$. The action $P'\rightarrow \operatorname{GL}_2(\mathbb{Z})\cong\operatorname{Aut}(\Lambda)$ is given by the integral representation $\rho:a_1 \mapsto \begin{pmatrix}1 & 0 \\ 0 &-1\end{pmatrix}$, over which the tensor product $\otimes_\rho$ is taken. Here, the spectral sequence degenerates at the fourth page, as such we are required to compute both $d_2$ and $d_3$.

Firstly, the cohomology ring of cyclic (and dihedral) groups have been computed [\onlinecite{Hill}], [\onlinecite{Dav}] 
\begin{equation}
H^\ast(C_n)\cong \mathbb{Z}[a_2]/\langle na_2\rangle\nonumber
\end{equation}
with a single degree-2 generator $a_2$. In particular we have a 2-periodicity $H^{p}(C_2)\cong H^{p+2}(C_2)$ given by $\cdot \cup a_2$ such that a generic class above degree $p=0$ can be written as
\begin{equation}
e_p = \begin{cases}0&; p=2j-1\\ a_2^j &; p=2j\end{cases}\in H^p(C_2),\qquad j\geq 1.\nonumber
\end{equation}
This also induces a 2-periodicity of the spectral sequence, and a generic class reads as
\begin{widetext}
\begin{equation} 
f_{p,q}^2=\begin{cases}0 &; p=2j-1\quad\text{or}\quad q>2 \\ \begin{cases}a_2^j &; q=0 \\  a_2^j\otimes t_1 - a_2^j\otimes t_2 &; q=1 \\ - a_2^j\otimes t_1\cup t_2 &; q=2 \end{cases} &; p=2j\end{cases}\in E^{p,q}_2,\nonumber
\end{equation}
\end{widetext}
where the explicit representation $\rho$ gives $a_2 \otimes_\rho t_1= a_2\otimes t_1$ but $a_2\otimes_\rho t_2 = -a_2\otimes t_2$. 

The differential $d_2$ is given by the Yoneda product with an extension class $\nu\in H^2(C_2,\Lambda)$. With the $\rho$-action explicit as above, it is the cup product $\cdot \cup \nu:E^{p,q}_2\rightarrow H^{p+2}(P',\Lambda\otimes H^q(\Lambda))$ composed with the slant product $\cdot\setminus\cdot:\Lambda\otimes H^q(\Lambda) \rightarrow H^{q-1}(\Lambda)$ induced on the coefficients, satisfying $t_i\setminus t_j = \delta_{ij}$. As the only non-trivial value of the glide class $\nu$ is the primitive vector $t_1$, a contraction with $\nu$ under the slant product kills $t_2$. 

This gives, in particular, $d_2(f_{2j,2}) = -a_2^{j+1}\otimes t_2$, hence $d_2$ there is injective (so $E^{p,2}_3=0$) with its image spanned by $t_2$. Correpondingly, $d_2$ is surjective at the $q=1$-st row with kernel spanned by $t_2$, so the third page completely collapses $E_3^{p,q}=0$ away from the $p=0$-th column. At $p=0$ and $q=2$, we see that $d_2(t_1\cup t_2) =-a_2\otimes t_2$, hence the kernel there is in fact $2\mathbb{Z}=E^{0,2}_3$. No other factors contribute to $H^2(\mathtt{pg})$, whence 
\begin{equation}
H^2(\mathtt{pg}) =\operatorname{gr}_2H^2(\mathtt{pg}) \cong 2\mathbb{Z}.\label{eq:glidecoh}
\end{equation}
It then follows that $H^2(\mathtt{pgg}) \cong 4\mathbb{Z}$.

\subsubsection{Cohomology of the 3D screw lattice}
Next let us turn to the 3D screw lattice $\Lambda=\mathbb{Z}^2\oplus\Lambda_A$, where $A$ is a divisor of the order $n=2,3,4,6$ of the rotational point group $P'=1\times C_n$, and the space group is classified by $H^2(C_n,\Lambda)\cong\mathbb{Z}/\frac{n}{A}$. For definiteness consider the case $n=3$ and $A=1$, with the representation 
\begin{equation}
\rho:(1,a_1)\mapsto \begin{pmatrix}1&0&0 \\ 0 & 0 & 1 \\ 0 &-1&-1\end{pmatrix}\in 1\oplus \operatorname{GL}_2(\mathbb{Z})\subset\operatorname{GL}_3(\mathbb{Z}) \nonumber
\end{equation} 
sending a generator $a_1\in C_3$ into a 2D rotation matrix by the axial property of $P'$. A generic class $e_p\in H^p(C_3)$ takes the same form as previously, but now the degree-2 generator $a_2$ has order-3 and we have four non-trivial rows in $E_2$
\begin{widetext}
\begin{equation}
f_{p,q}^2(B) = B\begin{cases}0&; p=2j-1\quad \text{or}\quad q>3 \\ 
\begin{cases} a_2^j &; q=0 \\ \sum\limits_{\alpha\leq 3}a_2^j\otimes_\rho t_\alpha &; q=1 \\ \sum\limits_{\alpha<\beta}a_2^j\otimes_\rho t_\alpha\cup t_\beta &; q=2 \\ a^j_2\otimes_\rho t_1\cup t_2\cup t_3 &; q=3\end{cases} &; p=2j\end{cases} \in E^{p,q}_2,\qquad B\in\mathbb{Z}/3.\nonumber
\end{equation}
\end{widetext}
Following the above, we first decompose $\otimes_\rho$ into $\otimes$. The representation $\rho$ induces
\begin{widetext}
\begin{equation}
f_{2j,1}^2(B) = \begin{cases}0 &; B= 0\mod 3 \\ Ba_2^j\otimes (t_1-t_2) &; B=1\mod 3 \\ Ba_2^j \otimes (t_1-t_3) &; B=2\mod 3\end{cases},\qquad f_{2j,2}^2(B) = \begin{cases}0 &; B=0\mod 3 \\ -Ba_2^j\otimes (t_1+t_3)\cup t_2 &; B=1\mod 3 \\ -Ba_2^j\otimes (t_1-t_2)\cup t_3 &; B=2\mod 3\end{cases}, \nonumber
\end{equation}
\end{widetext}
while the $q=0,3$ rows are invariant under the replacement $\otimes_\rho\rightarrow\otimes$. Now the extension class $\nu\in H^2(1\times C_3,\Lambda)$ it takes values along $t_1$, the Yoneda product against $\nu$ kills $t_2$ and $t_3$. For definiteness suppose $\nu=1$ is the generator, we then have $d_2(f_{2j,3}(B)) = Ba_2^{j+1}\otimes t_2\cup t_3$, so $d_2$ is injective at $q=3$ and $E_3^{p,3}=0$ away from the $p=0$-th column.

Next, we have $d_2(f_{2j,2}(B)) = -Ba_2^{j+1}\otimes t_{B+1}$, hence its kernel is spanned by the plaquette $t_{2}\cup t_{3}$. This coincides precisely with the image of $d_2$ coming from the row above, hence the $q=2$-nd row is also killed $E^{p,2}_3=0$ at the third page away from $p=0$. Lastly, we have $d_2(f_{2j,1}(B)) = Ba_2^{j+1}$, and its kernel is spanned by the vector $t_{B+1}$. This again coincides with the image from the row above, hence $q=1$-st row also dies $E^{p,1}_3=0$. Furthermore, here $d_2$ is surjective, thus the $q=0$-th row is gone as well, reduced to $E^{p,0}_3=0$.

All that remains, as in the previous case, is the $p=0$-th column. The map $d_2(m t_1\cup t_2\cup t_3)=m a_2\otimes t_3\cup (-t_2-t_3) = ma_2 t_2\cup t_3$ at the top $q=3$-rd row has kernel $m\in 3\mathbb{Z}$. No other factors contribute to $H^3(P)$, hence
\begin{equation}
H^3(P) = \operatorname{gr}_3H^3(P) = E^{0,3}_3 \cong 3\mathbb{Z}.\nonumber
\end{equation}
It follows that $H^3(P) \cong n\mathbb{Z}$ for screw axes $\Lambda=\mathbb{Z}^2\oplus\Lambda_A$ characterized by $A=1$ for any $n$-fold rotational symmetry.

If $A\neq 1$, however, then the entries $f_{p,q}^2$ exhibit order $A$. If $\nu=1$ remains the generator of $\mathbb{Z}/\frac{n}{A}$, then the above computations carry through and the only objects surviving to $E_3$ is once again concentrated on the $p=0$-th column. To compute the kernel $\operatorname{ker}d_2 = E^{0,3}_3=\operatorname{gr}_0H^3(P)$, we partition $n\mathbb{Z}$ into $n$-divisor classes characterized by $A$: namely $x\sim y\in n\mathbb{Z}$ iff they are the same multiplie of the divisor $\frac{n}{A}$ of $n$. Clearly, there are $A$ number of such divisor classes, whence 
\begin{equation}
H^3(P) = \frac{n}{A}\mathbb{Z}\oplus\mathbb{Z}/A.\label{eq:screwcoh}
\end{equation}
Now if $\nu$ generates a subset in $\mathbb{Z}/\frac{n}{A}$, then the differential $d_2$ can modify the torsion of (or even kill) the entries $f_{p,q}^2$. This can only occur if $\frac{n}{A}$ is not prime, however, which is impossible in the standard cases $n=2,3,4,6$ and $A\neq 1$.

\subsubsection{Cohomology of the 3D glide-screw lattice}
Let us now turn to the chirorotational case $P' = C_{nh}=C_2\times C_n$ with a glide-screw dislocated lattice $\Lambda$ and space group $P=\mathtt{gs}$. It is characterized by the extension class $\nu=\nu_1 \oplus \nu_2 \in H^2(P',\Lambda) = \mathbb{Z}/2\oplus\mathbb{Z}/\frac{n}{A'}$ for which $\frac{n}{A'}$ is even. By Eilenberg-Zilber we assemble
\begin{eqnarray}
H^\ast(C_{nh}) &\cong& H^\ast(C_2)\otimes H^\ast(C_n) \nonumber\\
&=& \mathbb{Z}[r_2]/\langle 2r_2\rangle\otimes \mathbb{Z}[a_2]/\langle na_2\rangle, \nonumber
\end{eqnarray}
from which $r_2^j\otimes a_2^j$ generates $H^{2j}(C_2\times C_n)$. Given the unique glide-screw configuration $n=6,A'=3$ for definiteness, consider the representation $\rho=\rho_1\oplus\rho_2$ defined by
\begin{equation}
(\alpha,a_1)\mapsto \begin{pmatrix}-1 &0 &0 \\ 0 & 0 & 1 \\ 0 & -1 &-1\end{pmatrix} \in \operatorname{GL}_1(\mathbb{Z})\oplus\operatorname{GL}_2(\mathbb{Z}). \nonumber
\end{equation}
We have four non-trivial rows in $E_2$:
\begin{widetext}
\begin{equation}
f_{p,q}^2(B) = \begin{cases} 0 &; p=2j-1\quad\text{or}\quad q>3 \\ \begin{cases}r_2^j\otimes Ba_2^j &; q=0 \\ \sum\limits_{\alpha\leq 3}(r_2^j\otimes Ba_2^j)\otimes_\rho t_\alpha &; q=1 \\ \sum\limits_{\alpha<\beta}(r_2^j\otimes Ba_2^j)\otimes_\rho t_\alpha\cup t_\beta &; q=2 \\ (r_2^j\otimes Ba_2^j)\otimes_\rho t_1\cup t_2\cup t_3 &; q=3\end{cases} &; p=2j\end{cases}\in E^{p,q}_2,\qquad B\in\mathbb{Z}/6. \nonumber
\end{equation}
\end{widetext}
for which we once again decompose $\otimes_\rho=\otimes_{\rho_1\oplus\rho_2}$ into $\otimes$. Toward this, we write $b=B\mod 2$ so that 
\begin{widetext}
\begin{eqnarray}
f_{2j,1}^2(B) &=& \begin{cases}\delta_{B,3}\sum\limits_{\alpha\leq 3}(r_2^j\otimes Ba_2^j)\otimes t_\alpha &; B\mod 3=0 \\ (r_2^j\otimes Ba_2^j)\otimes ((-1)^{b+1}t_1-t_2) &; B\mod 3=1 \\ (r_2^j\otimes Ba_2^j) \otimes ((-1)^{b+1}t_1-t_3) &; B \mod 3=2\end{cases},\nonumber \\ 
f_{2j,2}^2(B) &=& \begin{cases}\delta_{B,3}\sum\limits_{\alpha<\beta}(r_2^j\otimes Ba_2^j)\otimes t_\alpha\cup t_\beta &; B\mod 3=0 \\ -(r_2^j\otimes Ba_2^j)\otimes ((-1)^{b+1}t_1+t_3)\cup t_2 &; B\mod 3=1 \\ -(r_2^j\otimes Ba_2^j)\otimes ((-1)^{b+1}t_1-t_2)\cup t_3 &; B\mod 3 =2\end{cases}, \nonumber\\
f_{2j,3}^2(B) &=& (-1)^{b+1}(r_2^j\otimes Ba_2^j)\otimes t_1\cup t_2\cup t_3,\nonumber
\end{eqnarray}
\end{widetext}
with the $q=1$-st row remaining invariant.

The fixed points of $C_{nh}$ consists of the axis $\Lambda_2\subset\Lambda$, hence the extension class $\nu=\nu_1\oplus\nu_2$ takes values in the span of $t_2$. The differentials are once again injective (resp. surjective) at $q=3$ (resp. $q=1$) for $p>0$, so $E_3^{p,3}=E^{p,0}_3=0$. The images there are given by 
\begin{eqnarray}
d_2(f_{2j,3}^2(B)) &=& (-1)^{b+1}(r_2^{j+1}\otimes Ba_2^{j+1})\otimes t_1\cup t_3\nonumber\\
d_2(f_{2j,1}(B)) &=& (-1)^{b+1}\delta_{b,1}(r_2^{j+1}\otimes Ba_2^{j+1}).\nonumber
\end{eqnarray} 
The computation is less trivial in the middle columns. At $q=2$, none of the terms in $f_{p,q}^2(B)$ are killed by $d_2$ when $B\mod 3=1$, while its kernel in the other cases are spanned by $t_1\cup t_3$. As such the third page $E^{p,2}_3=0$ nevertheless vanishes. 

The image of $d_2$ at $q=2$ reads 
\begin{widetext}
$$d_2(f_{2j,2}(B)) = \begin{cases}\delta_{b,1}(r_2^{j+1}\otimes a_2^{j+1})\otimes (t_1+ t_3) &; B\mod 3 = 0 \\ -(r_2^{j+1}\otimes Ba_2^{j+1})\otimes ((-1)^{b+1}t_1+t_3) &; B\mod 3=1 \\   (r_2^{j+1}\otimes Ba_2^{j+1})\otimes t_3 &; B\mod 3=2 \end{cases}.$$ 
\end{widetext}
On the other hand, at $q=1$ notice that the kernel of $d_2$ is spanned by just $t_1$ when $B\mod 3 =1$, while it is spanned by both $t_1$ and $t_3$ otherwise. This means that 
\begin{eqnarray}
E^{2j,1}_3&=&\{0\}\cup \{f_{2j,1}^3(B)\mid B\mod 3=2\}\nonumber \\
&\cong&\mathbb{Z}/3\nonumber
\end{eqnarray} 
is in fact non-trivial. The cohomology $H^3(\mathtt{gs})$ then receives contribution from $\operatorname{gr}_2H^3(\mathtt{gs}) = E^{2,1}_3 =\mathbb{Z}/3$, as well as those form the $p=0$-th column $\operatorname{gr}_0H^3(\mathtt{gs}) = E^{3,0}$. Due to the precense of the order-2 generator $r_2$, the kernel of $d_2$ at $p=0,q=3$ is $2\mathbb{Z}$, hence we have
\begin{equation}
H^3(\mathtt{gs}) = \operatorname{gr}_0H^3(\mathtt{gs}) \oplus \operatorname{gr}_2H^3(\mathtt{gs}) \cong 2\mathbb{Z} \oplus\mathbb{Z}/3.\label{eq:gscoh}
\end{equation}

\subsubsection{Cohomology of the 3D tetrahedral glide lattice: computational sketch}
Let us now turn to the non-axial case. Consider the glide lattice corresponding to the tetrahedral point group $P'=S_4$, which acts on the tetrahedral lattice $\Lambda_t$ via the representation
\begin{equation}
\rho_t(a_1) = \begin{pmatrix}0&1&0\\1&0&0\\0&0&1\end{pmatrix},\qquad \rho_t(a_2) = \begin{pmatrix}1&0&0 \\ 0&0&1&\\ -1&-1&-1\end{pmatrix},\nonumber
\end{equation}
where $a_1,a_2\in S_4$ are the generators in the presentation Eq. (\ref{eq:sym4}). Several low-degree integral cohomology groups $H^n(S_4)$ have been previously computed, but so far the complete ring structure is not available. As such the spectral sequence computation must proceed entry-by-entry. We do not attempt this here, but only provide a sketch for the procedure.

It is possible to identify entries that should be relevant. As the LHS spectral sequence degenerates at page four at the latest, we see that
\begin{equation}
H^3(P) \cong \bigoplus_{p\leq 3}\operatorname{gr}_pH^{3-p}(P) = \bigoplus_{p\leq 3}E^{p,3-p}_4,\nonumber
\end{equation}
and one must compute the kernel of the differential $d_3$ at $(p,3-p)$ and its image from $(p+3,5-p)$. This in turn requires us to know the kernels of the differential $d_2$ at $(p,3-p)$ and $(p+3,5-p)$, as well as its image from $(p+2,4-p)$ and $(p+5,6-p)$. The classes we are required to assemble at the second page are therefore
\begin{equation}
f_{p,3-p}^2,\qquad f_{p+2,4-p}^2,\qquad f_{p+3,5-p}^2,\qquad f_{p+5,6-p}^2\nonumber
\end{equation}
for $1\leq p\leq 3$. This would require knowledge of the generators of $H^n(S_4)$ for $n$ from $1$ through $8$, at which point these classes can be explicitly written down using the give tetrahedral representation $\rho_t$.

Next, the cocycle representative $\nu$ for the glide class $\nu\in H^2(S_4,\Lambda_t)\cong\mathbb{Z}/2$ is valued on the axis $M_1\cap M_2\cap M_3 = \mathbb{Z}[t_g]\subset\Lambda_t$ along the intersection of the three reflection planes $M_1,M_2,M_3$. The differential $d_2$ is then a cup product with $\nu$ and a contraction against the the vector $t_g$, from which its kernel and image can be extracted in order to construct entries in the third page. 

The degree-3 differential $d_3$, however, is an extremely difficult object to explicitly write down. The situation would be significantly simplified if those entries away from the $p=0$-th column are all killed, but this is not guaranteed as the glide axis $t_g$ is a linear combination of the tetrahedral primitive vertex vectors $t_1,t_2,t_3$.

\bibliography{aipsamp}% Produces the bibliography via BibTeX.

\nocite{*}

\end{document}